\documentclass[conference,10pt,twocolumns]{IEEEtran}

\usepackage{graphicx}
\usepackage[center]{caption}
\usepackage{epsfig,latexsym,amsmath,amsfonts}
\usepackage{mathtools}
\DeclareMathOperator*{\argmin}{argmin}
\usepackage{amssymb}
\usepackage{placeins}
\usepackage{subfigure}
\usepackage{verbatim}
\usepackage{cite,url}
\usepackage{epsf,bm}
\usepackage{color}
\usepackage{epstopdf}
\usepackage{romannum,hyperref}
\usepackage[onehalfspacing]{setspace}
\usepackage{dsfont}
\usepackage[utf8]{inputenc}
\usepackage[english]{babel}
\usepackage{amsthm}
\newtheorem{definition}{Definition}

\usepackage[utf8]{inputenc}
\usepackage[english]{babel}

\newtheorem{theorem}{Theorem}
\newtheorem{lemma}[theorem]{Lemma}
\newtheorem{proposition}[theorem]{Proposition}
\newtheorem{corollary}[theorem]{Corollary}

\usepackage[lined, boxed, linesnumbered, ruled]{algorithm2e}
\IEEEoverridecommandlockouts
\begin{document}

\title{Minimizing the Age of Information through Queues}

\author{\large Ahmed M. Bedewy, Yin Sun, \emph{Member, IEEE}, and Ness B. Shroff, \emph{Fellow, IEEE}
\thanks{This paper was presented in part at IEEE ISIT 2016 \cite{age_optimality_multi_server}.}
\thanks {This work has been supported in part by ONR grants N00014-17-1-2417 and N00014-15-1-2166, Army Research Office grants W911NF-14-1-0368 and MURI W911NF-12-1-0385,  National Science Foundation grants CNS-1446582, CNS-1421576, CNS-1518829, and CCF-1813050, and a grant from the Defense Thrust Reduction Agency HDTRA1-14-1-0058.}
\thanks{A. M. Bedewy is with the  Department  of  ECE,  The  Ohio  State  University, Columbus, OH 43210 USA (e-mail:  bedewy.2@osu.edu).}
\thanks{Y.  Sun  is  with  the  Department  of  ECE,  Auburn  University,  Auburn,  AL 36849 USA (e-mail:  yzs0078@auburn.edu).}
\thanks{N. B.  Shroff  is  with  the  Department  of  ECE and  the  Department  of  CSE, The Ohio State University,  Columbus, OH 43210 USA  (e-mail:  shroff.11@osu.edu).}

}
\maketitle
\begin{abstract}

In this paper, we investigate scheduling policies that minimize the age of information in single-hop queueing systems. We propose a Last-Generated, First-Serve (LGFS) scheduling policy, in which the packet with the earliest generation time is processed with the highest priority. If the service times are \emph{i.i.d.} exponentially distributed, the preemptive LGFS policy is proven to be age-optimal in a stochastic ordering sense. If the service times are \emph{i.i.d.} and satisfy a New-Better-than-Used (NBU) distributional property, the non-preemptive LGFS policy is shown to be within a constant gap from the optimum age performance. These age-optimality results are quite general: (i) They hold for arbitrary packet generation times and arrival times (including out-of-order packet arrivals), (ii) They hold for multi-server packet scheduling with the possibility of replicating a packet over multiple servers, (iii) They hold for minimizing not only the time-average age and mean peak age, but also for minimizing the age stochastic process and any non-decreasing functional of the age stochastic process. If the packet generation time is equal to packet arrival time, the LGFS policies reduce to the Last-Come, First-Serve (LCFS) policies. Hence, the age optimality results of LCFS-type policies are also established.

\end{abstract}
\section{Introduction}\label{Int}
The ubiquity of mobile devices and applications has greatly boosted the demand for real-time information updates, such as news, weather reports, email notifications, stock quotes, social updates, mobile ads, etc. Also, timely status updates are crucial in networked monitoring and control systems. These include, but are not limited to, sensor networks used to measure temperature or other physical phenomena, and surrounding monitoring in autonomous driving.

A common need in these real-time applications is to keep the destination (i.e., information consumer) updated with the freshest information. To identify the timeliness of the updates, a metric called the \emph{age-of-information}, or simply \emph{age}, was defined in, e.g., \cite{adelberg1995applying,cho2000synchronizing,golab2009scheduling,KaulYatesGruteser-Infocom2012}. At time $t$, if the packet with the largest generation time at the destination was generated at time $U(t)$, the age $\Delta(t)$ is defined as
\begin{equation}
\Delta(t)=t-U(t). 
\end{equation}
 Hence, age is the time elapsed since the freshest received packet was generated.


In recent years, a variety of approaches have been investigated to reduce the age. In \cite{KaulYatesGruteser-Infocom2012,2012ISIT-YatesKaul,2015ISITHuangModiano},  it was found in First-Come, First-Serve (FCFS) queueing systems that the time-average age first decreases with the update frequency and then increases with the update frequency. The optimal update frequency was obtained to minimize the age in FCFS systems. In \cite{CostaCodreanuEphremides_TIT,CostaCodreanuEphremides2014ISIT,Icc2015Pappas}, it was shown that the age can be further improved by discarding old packets waiting in the queue when a new sample arrives. Characterizing the age in Last-Come, First-Serve (LCFS) queueing systems with gamma distributed service times was considered in \cite{Gamma_dist}. However, these studies cannot tell us (i) which queueing discipline can minimize the age and (ii) under what conditions the minimum age is achievable.


\begin{figure}
\includegraphics[scale=0.35]{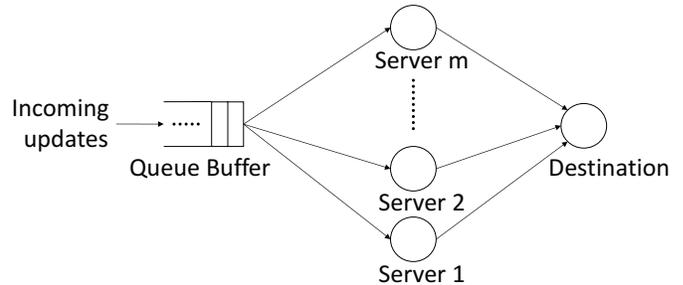}
\centering
\caption{System model.}\label{Fig:sysMod}
\vspace{-0.3cm}
\end{figure}


In this paper, we answer these two questions for an information-update system illustrated in Fig. \ref{Fig:sysMod}, where a sequence of update packets arrive at a queue with $m$ servers and a buffer size $B$. Each server can be used to model a channel in multi-channel communication systems \cite{Shroff_multichannel}, or a computer in parallel computing systems \cite{kumar1994introduction}. The service times of the update packets are \emph{i.i.d.} across servers and the packets assigned to the same server. Let $s_i$ be the generation time of the update packet $i$ at an external source, and $a_i$ be the arrival time of the update packet $i$ at the queue. Out-of-order packet arrivals are allowed, such that the packets may arrive in an order different from their generation times, e.g., $s_i < s_j$ but $a_j < a_i$. Packet replication \cite{Chen_queue_coding,yin_data_retrieve,sun2016delay} is considered in this study. In particular, multiple replicas of a packet can be assigned to different servers, at possibly different service starting time epochs. The first completed replica is considered as the valid execution of the packet; after that, the remaining replicas of this packet are cancelled immediately to release the servers. Suppose that a packet can be replicated on at most $r$ servers ($r\leq m$), where $r$ is called the maximum replication degree. If $r = 1$, this reduces to the case where replication is not allowed at all.  We propose a Last-Generated, First-Serve (LGFS) scheduling policy, in which the packet with the earliest generation time is served with the highest priority. The following are the key contributions of this paper:
\begin{itemize}
\item If the packet service times are \emph{i.i.d.} exponentially distributed, then for \emph{arbitrary} system parameters (including \emph{arbitrary} packet generation times $s_i$, packet arrival times $a_i$, number of servers $m$, maximum replication degree $r$, and buffer size $B$), we prove that the preemptive LGFS with replication (prmp-LGFS-R) policy minimizes the age stochastic process  and any non-decreasing functional of the age stochastic process among all policies in a stochastic ordering sense (Theorem \ref{thm1}). Note that this age penalty model is very general. Many age penalty metrics studied in the literature, such as the time-average age  \cite{KaulYatesGruteser-Infocom2012,RYatesTIT16,KamKompellaEphremidesTIT,Bo_pull_model,CostaCodreanuEphremides2014ISIT,2012ISIT-YatesKaul,CostaCodreanuEphremides_TIT,Icc2015Pappas,BacinogCeranUysal_Biyikoglu2015ITA,2015ISITYates,Gamma_dist}, average peak age \cite{CostaCodreanuEphremides2014ISIT,2015ISITHuangModiano,CostaCodreanuEphremides_TIT,BacinogCeranUysal_Biyikoglu2015ITA,Gamma_dist,age_with_delivery_error}, and time-average age penalty function \cite{generat_at_will,SunJournal2016}, are special cases of this age penalty model. 

\item We further investigate a more general class of packet service time distributions called New-Better-than-Used (NBU) distributions. We show that the non-preemptive Last-Generated, First-Serve with replication (non-prmp-LGFS-R) policy is within a constant age gap from the optimum average age, and that the gap is independent of the system parameters mentioned above (Theorem \ref{thmNBU1}). Note that  policy non-prmp-LGFS-R with a maximum replication degree $r$ can be near age-optimal compared with policies with any maximum replication degree. This result was not anticipated: In \cite{sun2016delay, shah2016redundant, joshi2017efficient}, it was shown that non-replication policies are near delay-optimal and replication policies are far from the optimum delay and throughput performance for NBU service time distributions. From these studies, one would expect that replications may worsen the age performance. To our surprise, however, we found that a replicative policy (i.e., non-prmp-LGFS-R) is near-optimal in minimizing the age, even for NBU service time distributions.

\item For a special case of the system settings where the update packets arrive in the same order of their generation times and there is no replication, the prmp-LGFS-R policy reduces to LCFS with preemption in service for a single source case in \cite{RYatesTIT16}, and the non-prmp-LGFS-R when $B=1$ reduces to LCFS with preemption only in waiting for a single source case in \cite{RYatesTIT16}, or the ``M/M/1/2*'' in \cite{CostaCodreanuEphremides_TIT,CostaCodreanuEphremides2014ISIT}. Hence, our optimality results are also established for these LCFS-type policies. This relationship tells us that this policy can achieve age-optimality in this case. 

\item Finally, we investigate the throughput and delay performance of the proposed policies. We show that if the packet service times are i.i.d. exponentially distributed, then the prmp-LGFS-R policy is also throughput and delay optimal among all policies (Theorem \ref{thm2th}). In addition, if the packet service times are \emph{i.i.d.} NBU and replications are not allowed, then the non-prmp-LGFS policy is throughput and delay optimal among all non-preemptive policies (Theorem \ref{thmNBU1th}). 
\end{itemize}

To the best of our knowledge, these are the first optimality results on minimizing the age-of-information in queueing systems. Moreover, this is the first paper that considers packet replication to minimize the age.



The remainder of this paper is organized as follows. After a brief overview of related work in Section \ref{RW}, we present the model and problem formulation in Section \ref{sysmod}. The age of the proposed policies is analyzed in Section \ref{age_result}, and the throughput and delay performance of these policies are investigated in Section \ref{thro-anal}. Finally, we conclude in Section \ref{Concl}.

\section{Related Work}\label{RW}
A series of works studied the age performance of scheduling policies in a single queueing system with Poisson arrival process and exponential service time \cite{KaulYatesGruteser-Infocom2012,2012ISIT-YatesKaul,CostaCodreanuEphremides_TIT,RYatesTIT16,CostaCodreanuEphremides2014ISIT,Icc2015Pappas,KamKompellaEphremidesTIT}. In \cite{KaulYatesGruteser-Infocom2012,2012ISIT-YatesKaul}, the update frequency was optimized to improve data freshness in FCFS information-update systems. The effect of the packet management on the age was considered in \cite{CostaCodreanuEphremides_TIT,CostaCodreanuEphremides2014ISIT,Icc2015Pappas}. It was found that a good policy is to discard the old updates waiting in the queue when a new sample arrives, which can greatly reduce the impact of queueing delay on data freshness. In \cite{RYatesTIT16}, the time-average age was characterized for multiple sources Last-Come, First-Serve (LCFS) information-update systems with and without preemption. In this study, it was shown that sharing service facility among Poisson sources can improve the total age. Characterizing the time average age for FCFS queueing system with two and infinite number of servers was studied in \cite{KamKompellaEphremidesTIT}. The analysis in \cite{KamKompellaEphremidesTIT} showed that the model with infinite servers has a lower age in conjunction with more wasting of network resources due to the rise in the obsolete delivered packets. One open question in these studies on age analysis \cite{KaulYatesGruteser-Infocom2012,2012ISIT-YatesKaul,CostaCodreanuEphremides_TIT,RYatesTIT16,CostaCodreanuEphremides2014ISIT,Icc2015Pappas,KamKompellaEphremidesTIT}  is whether the preemptive LCFS policy is age-optimal for exponential service times. In this paper, we provide a confirmative answer to this question, and further investigate age-optimality for more general system settings such as arbitrary packet generation and arrival processes (including out-of-order packet arrivals), multi-server networks, as well as packet replications over multiple servers.


In  \cite{Bo_pull_model}, the average age was characterized in a pull model, where a customer sends requests to all servers to retrieve (pull) the interested information. In this model, the servers carry information with different freshness level and a user waits for the responses from these servers. The server updating process and the response times were assumed to be Poisson and exponential, respectively. In contrast with \cite{Bo_pull_model}, where the authors assumed that a user contacts servers to check for updates, here we prove age-optimality in a multi-server queueing system where a user sends the  updates to a destination through the servers and packet replication is considered.


 Characterizing the age for a class of packet service time distributions that are more general than exponential was considered in \cite{2015ISITHuangModiano,age_with_delivery_error,Gamma_dist}. In \cite{2015ISITHuangModiano}, the age was analyzed in multi-class M/G/1 and M/G/1/1 queues. The age performance in the presence of errors when the service times are exponentially distributed was analyzed in  \cite{age_with_delivery_error}. Gamma-distributed service times was considered in \cite{Gamma_dist}. The studies in \cite{age_with_delivery_error,Gamma_dist} were carried out for LCFS queueing systems with and without preemption. In complement with the age analysis results in \cite{2015ISITHuangModiano,age_with_delivery_error,Gamma_dist}, we show that non-preemptive LGFS (and its special case non-preemptive LCFS) policies are near age-optimal for NBU service time distributions. Similar to the exponential case, these results for NBU service times hold for arbitrary packet general and arrival processes, multiple server networks, and packet replication over multiple servers. In addition, gamma distribution considered in \cite{age_with_delivery_error,Gamma_dist} is a special case of NBU service time distributions.




 
%
In our study, packet generation and arrival times are not controllable. Another line of research has been the joint optimization of packet generation and transmissions in \cite{BacinogCeranUysal_Biyikoglu2015ITA,2015ISITYates,generat_at_will,SunJournal2016}.  An information update policy was developed in \cite{generat_at_will,SunJournal2016}, which was proven to minimize a general class of non-negative,
non-decreasing age penalty functions among all causally feasible policies. 
More recently, a real-time sampling problem of the Wiener process has been studied in \cite{Sun_reportISIT17}: If the sampling times are independent of the observed Wiener process, the optimal sampling problem in \cite{Sun_reportISIT17} reduces to an age-of-information optimization problem; otherwise, the optimal sampling policy can use knowledge of the Wiener process to achieve better performance than age-of-information optimization.

Recently, we generalized our results to multihop networks in \cite{multihop_optimal}, where we proved that age-optimality is achievable in multihop networks with arbitrary packet generation times, packet arrival times, and general network topologies. In particular, it was shown that the LGFS policy is age-optimal among all causal policies for exponential packet service times. In addition, for arbitrary distributions of packet service times, it was shown that the LGFS policy is age-optimal among all non-preemptive work-conserving policies. 

The considered age penalty model in this paper is very general such that it includes, but is not limited to, the time-average age  \cite{KaulYatesGruteser-Infocom2012,RYatesTIT16,KamKompellaEphremidesTIT,Bo_pull_model,CostaCodreanuEphremides2014ISIT,2012ISIT-YatesKaul,CostaCodreanuEphremides_TIT,Icc2015Pappas,BacinogCeranUysal_Biyikoglu2015ITA,2015ISITYates,Gamma_dist}, average peak age \cite{CostaCodreanuEphremides2014ISIT,2015ISITHuangModiano,CostaCodreanuEphremides_TIT,BacinogCeranUysal_Biyikoglu2015ITA,Gamma_dist,age_with_delivery_error}, and time-average age penalty function \cite{generat_at_will,SunJournal2016}.


\section{Model and Formulation}\label{sysmod}
\subsection{Notations and Definitions}
For any random variable $Z$ and an event $A$, let $[Z\vert A]$ denote a random variable with the conditional distribution of $Z$ for given $A$, and $\mathbb{E}[Z\vert A]$ denote the conditional expectation of $Z$ for given $A$. 

Let $\mathbf{x}=(x_1,x_2,\ldots,x_n)$ and $\mathbf{y}=(y_1,y_2,\ldots,y_n)$ be two vectors in $\mathbb{R}^n$, then we denote $\mathbf{x}\leq\mathbf{y}$ if $x_i\leq y_i$ for $i=1,2,\ldots,n$. We use $x_{[i]}$ to denote the $i$-th largest component of vector $\mathbf{x}$. A set $U\subseteq \mathbb{R}^n$ is called upper if $\mathbf{y}\in U$ whenever $\mathbf{y}\geq\mathbf{x}$ and $\mathbf{x}\in U$. We will need the following definitions: 
\begin{definition} \textbf{ Univariate Stochastic Ordering:} \cite{shaked2007stochastic} Let $X$ and $Y$ be two random variables. Then, $X$ is said to be stochastically smaller than $Y$ (denoted as $X\leq_{\text{st}}Y$), if
\begin{equation*}
\begin{split}
\mathbb{P}\{X>x\}\leq \mathbb{P}\{Y>x\}, \quad \forall  x\in \mathbb{R}.
 \end{split}
\end{equation*}
\end{definition}
\begin{definition}\label{def_2} \textbf{Multivariate Stochastic Ordering:} \cite{shaked2007stochastic} 
Let $\mathbf{X}$ and $\mathbf{Y}$ be two random vectors. Then, $\mathbf{X}$ is said to be stochastically smaller than $\mathbf{Y}$ (denoted as $\mathbf{X}\leq_\text{st}\mathbf{Y}$), if
\begin{equation*}
\begin{split}
\mathbb{P}\{\mathbf{X}\in U\}\leq \mathbb{P}\{\mathbf{Y}\in U\}, \quad \text{for all upper sets} \quad U\subseteq \mathbb{R}^n.
 \end{split}
\end{equation*}
\end{definition}
\begin{definition} \textbf{ Stochastic Ordering of Stochastic Processes:} \cite{shaked2007stochastic} Let $\{X(t), t\in [0,\infty)\}$ and $\{Y(t), t\in[0,\infty)\}$ be two stochastic processes. Then, $\{X(t), t\in [0,\infty)\}$ is said to be stochastically smaller than $\{Y(t), t\in [0,\infty)\}$ (denoted by $\{X(t), t\in [0,\infty)\}\leq_\text{st}\{Y(t), t\in [0,\infty)\}$), if, for all choices of an integer $n$ and $t_1<t_2<\ldots<t_n$ in $[0,\infty)$, it holds that
\begin{align}\label{law9'}
\!\!\!(X(t_1),X(t_2),\ldots,X(t_n))\!\leq_\text{st}\!(Y(t_1),Y(t_2),\ldots,Y(t_n)),\!\!
\end{align}
where the multivariate stochastic ordering in \eqref{law9'} was defined in Definition \ref{def_2}.
\end{definition}

\subsection{Preliminary Propositions}
The following propositions will be used throughout the paper:
\begin{proposition}[\cite{shaked2007stochastic}, Theorem 6.B.3]\label{Theorem_6.B.3}
Let $\mathbf{X}=(X_1, X_2,\ldots, X_n)$ and $\mathbf{Y}=(Y_1, Y_2,\ldots, Y_n)$ be two $n$-dimensional random vectors. If
\begin{equation*}
X_1\leq_\text{st} Y_1,
\end{equation*}
\begin{equation*}
[X_2\vert X_1=x_1]\leq_\text{st}[Y_2\vert Y_1=y_1]~\text{whenever}~x_1\leq y_1,
\end{equation*}
and in general, for $i=2,3,\ldots,n$,
\begin{align*}
\begin{split}
[X_i\vert X_1=x_1,\ldots, X_{i-1}=x_{i-1}]\leq_\text{st} \\  [Y_i\vert Y_1=y_1,\ldots, Y_{i-1}=y_{i-1}]\\
\text{whenever}~x_j\leq y_j,~j=1,2,\ldots, i-1,
\end{split}
\end{align*}
then $\mathbf{X}\leq_\text{st}\mathbf{Y}$.
\end{proposition}

\begin{proposition}[\cite{shaked2007stochastic}, Theorem 6.B.16.(a)]\label{Theorem_6.B.16.(a)}
Let $\mathbf{X}$ and $\mathbf{Y}$ be two $n$-dimensional random vectors. If $\mathbf{X}\leq_\text{st}\mathbf{Y}$ and $\mathbf{q}:\mathbb{R}^n\rightarrow\mathbb{R}^k$ is any $k$-dimensional increasing [decreasing] function, for any positive integer $k$, then the $k$-dimensional vectors $\mathbf{q(X)}$ and $\mathbf{q(Y)}$ satisfy $\mathbf{q(X)}\leq_\text{st}[\geq_\text{st}]\mathbf{q(Y)}$.
\end{proposition}

\begin{proposition}[\cite{shaked2007stochastic}, Theorem 6.B.16.(b)]\label{Theorem_6.B.16.(b)}
Let $\mathbf{X}_1, \mathbf{X}_2,\ldots\mathbf{X}_d$ be a set of independent random vectors where the dimension of $\mathbf{X}_i$ is $k_i$, $i=1, 2,\ldots,d$. Let $\mathbf{Y}_1, \mathbf{Y}_2,\ldots\mathbf{Y}_d$ be another set of independent random vectors where the dimension of $\mathbf{Y}_i$ is $k_i$, $i=1, 2,\ldots,d$. Denote $k=k_1+k_2+\ldots+k_d$. If $\mathbf{X}_i\leq_\text{st}\mathbf{Y}_i$ for $i=1,2,\ldots,d$, then, for any increasing function $\psi:\mathbb{R}^k\rightarrow\mathbb{R}$, one has
\begin{equation*}
\psi(\mathbf{X}_1, \mathbf{X}_2,\ldots\mathbf{X}_d)\leq_\text{st}\psi(\mathbf{Y}_1, \mathbf{Y}_2,\ldots\mathbf{Y}_d).
\end{equation*}
\end{proposition}
\begin{proposition}[\cite{shaked2007stochastic}, Theorem 6.B.16.(e)]\label{Theorem_6.B.16.(e)}
Let $\mathbf{X}, \mathbf{Y}$, and $\mathbf{\Theta}$ be random vectors such that $[\mathbf{X}\vert\mathbf{\Theta}=\mathbf{\theta}]\leq_\text{st}[\mathbf{Y}\vert\mathbf{\Theta}=\mathbf{\theta}]$ for all $\mathbf{\theta}$ in the support of $\mathbf{\Theta}$. Then $\mathbf{X}\leq_\text{st}\mathbf{Y}$.
\end{proposition}
In the next proposition, $=_\text{st}$ denotes equality in law.
\begin{proposition}[\cite{shaked2007stochastic}, Theorem 6.B.30]\label{Theorem_6.B.30}
The random processes $\{X(t), t\in [0,\infty)\}$ and $\{Y(t), t\in [0,\infty)\}$ satisfy $\{X(t), t\in [0,\infty)\}\leq_\text{st}\{Y(t), t\in [0,\infty)\}$ if, and only if, there exist two random
processes $\{\widetilde{X}(t), t\in [0,\infty)\}$ and  $\{\widetilde{Y}(t), t\in [0,\infty)\}$, defined on the same probability space, such that
\begin{equation*}
\begin{split}
&\{\widetilde{X}(t), t\in [0,\infty)\}=_\text{st}\{X(t), t\in [0,\infty)\},\\
&\{\widetilde{Y}(t), t\in [0,\infty)\}=_\text{st}\{Y(t), t\in [0,\infty)\},
\end{split}
\end{equation*}
and
\begin{equation*}
\mathbb{P}\{\widetilde{X}(t)\leq\widetilde{Y}(t), t\in[0,\infty)\}=1.
\end{equation*}
\end{proposition}

\subsection{Queueing System Model}
We consider a queueing system with $m$ servers as shown in Fig. \ref{Fig:sysMod}. The system starts to operate at time $t=0$. The update packets are generated exogenously to the system and then arrive at the queue. Thus, the update packets may not arrive at the queue instantly when they are generated. The $i$-th update packet, called packet $i$, is generated at time $s_i$, arrives at the queue at time $a_i$, and is delivered to the destination at time $c_i$ such that $0\leq s_1\leq s_2\leq \ldots$ and $s_i\leq a_i\leq c_i$.
Note that in this paper, the sequences $\{s_1, s_2, \ldots\}$ and $\{a_1, a_2, \ldots\}$ are \textit{arbitrary}. Hence, the update packets may not arrive at the system in the order of their generation times.
For example, in Fig. \ref{Fig:Age}, we have $s_1<s_2$ but $a_2<a_1$. Let $B$ denote the buffer size of the queue which can be infinite, finite, or even zero. If $B$ is finite, the packets that arrive to a full buffer are either dropped or replace other packets in the queue.
The packet service times are \emph{i.i.d.} across servers and the packets assigned to the same server, and are independent of the packet generation and arrival processes. Packet replication is considered in this model, where the maximum replication degree is $r$ ($1\leq r\leq m$). In this model, one packet can be replicated to at most $r$ servers and the first completed replica is considered as the valid execution of the packet. After that, the remaining replicas of this packet are cancelled immediately to release the servers. Note that, the maximum replication degree $r$ is fixed for a system; however, the number of replicas that can be created for a packet may vary between 1 and $r$.

\subsection{Scheduling Policy}\label{Schpolicy}
A scheduling policy, denoted by $\pi$, determines the packet assignments and replications over time; it also controls dropping or replacing packets when the queue buffer is full. Note that the packet delivery time to the destination $c_i$ is a function of the scheduling policy $\pi$, while the sequences $\{s_1, s_2, \ldots\}$ and $\{a_1, a_2, \ldots\}$ do not change according to the scheduling policy. However, a policy $\pi$ may have knowledge of the future packet generation and arrival times. 
Moreover, we assume that the packet service times are invariant of the scheduling policy and the realization of a packet service time is unknown until its service is completed (unless the service time is deterministic).
 
Define $\Pi_r$ as the set of all policies, that includes \emph{causal} and \emph{non-causal} policies, when the maximum replication degree is $r$. Hence,  $\Pi_1\subset\Pi_2\subset\ldots\subset\Pi_m$. Note that causal policies are those policies whose scheduling decisions are made based only on the history and current state of the system; while non-causal policies are those policies whose scheduling decisions are made based on the history, current, and future state of the system. We define several types of policies in $\Pi_r$:

A policy is said to be \textbf{preemptive}, if a server can preempt a packet being processed and switch to processing  any other (including the preempted packet itself) packet at any time; only one copy of the preempted packet can be stored back into the queue if there is enough buffer space and sent at a later time when the servers are available again \footnote{If a preempted packet is served again, its service either starts over or it resumes the service from the preempted point. In case of exponential service times, both scenarios are equivalent because of the memoryless property.}. In contrast, in a \textbf{non-preemptive} policy, processing of a packet cannot be interrupted until the packet is completed or cancelled \footnote{Recall that a packet is cancelled when a replica has completed processing at another server.}; after completing or cancelling a packet, the server can switch to process another packet. A policy is said to be \textbf{work-conserving}, if no server is idle whenever there are packets waiting in the queue.
 
\subsection{Age Performance Metric}
\begin{figure}
\includegraphics[scale=0.35]{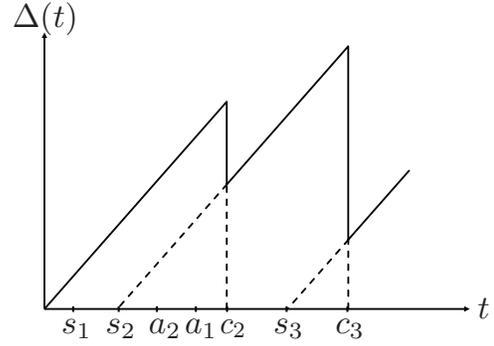}
\centering
\caption{Sample path of the age process $\Delta(t)$.}\label{Fig:Age}
\vspace{-0.3cm}
\end{figure}
Let $U(t)=\max\{s_i : c_{i}\leq t\}$ be the largest generation time of the packets at the destination at time $t$. The \emph{age-of-information}, or simply the \emph{age}, is defined as  \cite{adelberg1995applying,cho2000synchronizing,golab2009scheduling,KaulYatesGruteser-Infocom2012}
\begin{equation}\label{age}
\begin{split}
\Delta(t)=t-U(t).
\end{split}
\end{equation} 
The initial state $U(0^-)$ at time $t=0^-$ is invariant of the policy $\pi\in\Pi_r$, where we assume that $s_0=U(0^-)=0$. As shown in Fig. \ref{Fig:Age}, the age increases linearly with $t$ but is reset to a smaller value with the arrival of a packet with larger generation time. The age process is given by 
\begin{align}
\Delta=\{\Delta(t), t\in [0,\infty)\}.
\end{align}

In this paper, we introduce a non-decreasing \emph{age penalty functional} $g(\Delta)$ to represent the level of dissatisfaction for data staleness at the receiver or destination.
\begin{definition} \textbf{Age Penalty Functional:}
Let $\mathbf{V}$ be the set of $n$-dimensional Lebesgue measurable functions, i.e.,
\begin{align}
\mathbf{V} = \{f : [0,\infty)^n \mapsto \mathbb{R}~\text{such that} ~f~ \text{is Lebesgue measurable}\}.\nonumber
\end{align}
A functional $g:\mathbf{V}\mapsto\mathbb{R}$ is said to be an \emph{age penalty functional} if $g$ is \emph{non-decreasing} in the following sense:  
\begin{equation}
\begin{split}
\!\!g(\Delta_1) \leq g(\Delta_2),~\text{whenever}~\Delta_{1}(t)\leq \Delta_{2}(t), \forall t\in [0,\infty). \!\!
\end{split}
\end{equation}
\end{definition}
The age penalty functionals used in prior studies include:
\begin{itemize}
\item \emph{Time-average age  \cite{CostaCodreanuEphremides2014ISIT,KaulYatesGruteser-Infocom2012,2012ISIT-YatesKaul,CostaCodreanuEphremides_TIT,Icc2015Pappas,KamKompellaEphremidesTIT,Bo_pull_model,RYatesTIT16,BacinogCeranUysal_Biyikoglu2015ITA,2015ISITYates,Gamma_dist}:} The time-average age is defined as
\begin{equation}\label{functional1}
g_1(\Delta)=\frac{1}{T}\int_{0}^{T} \Delta(t) dt,
\end{equation}
\item \emph{Average peak age \cite{CostaCodreanuEphremides2014ISIT,2015ISITHuangModiano,CostaCodreanuEphremides_TIT,BacinogCeranUysal_Biyikoglu2015ITA,Gamma_dist,age_with_delivery_error}:} The average peak is defined as 
\begin{equation}\label{functional2}
g_2(\Delta)=\frac{1}{K}\sum_{k=1}^{K} A_{k},
\end{equation}
where $A_{k}$ denotes the $k$-th peak value of $\Delta(t)$ since time $t=0$. 
\item \emph{Time-average age penalty function \cite{generat_at_will,SunJournal2016}:} The average age penalty function is
\begin{equation}\label{functional3}
g_3(\Delta)= \frac{1}{T}\int_{0}^{T} h(\Delta(t)) dt,
\end{equation}
where $h$ : $[0,\infty)\to [0,\infty)$ can be any non-negative and non-decreasing function. As pointed out in  \cite{SunJournal2016}, a stair-shape function $h(x)=\lfloor x\rfloor$ can be used to characterize the dissatisfaction of data staleness when the information of interest is checked periodically, and an exponential function $h(x)=e^{x}$ is appropriate for online learning and control applications where the demand for updating data increases quickly with respect to the age. Also, an indicator function $h(x)=\mathds{1}(x>d)$ can be used to characterize the dissatisfaction when a given age limit $d$ is violated. 
\end{itemize}

\section{Age-Optimality Results of LGFS Policies}\label{age_result}
In this section, we provide age-optimality and near age-optimality results for multi-server queueing networks with packet replication. We start by considering the exponential packet service time distribution and show that age-optimality can be achieved. Then, we consider the classes of NBU packet service time distributions and show that there exist simple policies that can come close to age-optimality.
\subsection{Exponential Service Time Distribution}
We study age-optimal packet scheduling when the packet service times are \emph{i.i.d.} exponentially distributed. We start by defining the Last-Generated, First-Serve discipline as follows.
\begin{definition}
A scheduling policy is said to follow the \textbf{Last-Generated, First-Serve (LGFS)} discipline, if the last generated packet is served first among all packets in the system.
\end{definition}
In the LGFS disciplines, packets are served according to their generation times such that the packet with the largest generation time is served first among all packets in the system. In contrast, in the LCFS disciplines, packets are served according to their arrival times such that the packet with the largest arrival time is served first among all packets in the system. Both disciplines are equivalent when the packets arrive to the queue in the same order of their generation times.

In this paper, we propose a policy called \textbf{preemptive Last-Generated, First-Serve with replication (prmp-LGFS-R)}. This policy follows the LGFS discipline. When there is no replication ($r=1$), the implementation details of prmp-LGFS-R policy \footnote{The decision related to dropping or replacing packets in the full buffer case does not affect the age performance of prmp-LGFS-R policy. Hence, we don't specify this decision under the prmp-LGFS-R policy in all related algorithms.} are depicted in Algorithm \ref{alg1'}.
\begin{algorithm}[h]
\footnotesize
\SetKwData{NULL}{NULL}
\SetCommentSty{small} 
$\alpha:=0$\tcp*[r]{$\alpha$ is the smallest generation time of the packets under service}
$I:=m$\tcp*[r]{$I$ is the number of idle servers}
$Q:=\emptyset$\tcp*[r]{$Q$ is the set of distinct packets that are under service}
\While{the system is ON} {
\If{a new packet $p_i$ with generation time $s$ arrives}{ 
\uIf(\tcp*[f]{All servers are busy}){I=0}{
\uIf(\tcp*[f]{Packet $p_i$ is stale}){$s\leq\alpha$}{
Store the packet in the queue\;}
\Else(\tcp*[f]{Packet $p_i$ carries fresh information}){
Find packet $p_j\in Q$ with generation time $\alpha$\;
Preempt packet $p_j$ and store it back to the queue\;
Assign packet $p_i$ to the idle server\;
$Q:=Q\cup\{p_i\} - \{p_j\}$\;
}}
\Else(\tcp*[f]{At least one of the servers is idle})
{
Assign packet $p_i$ to an idle server\;
$Q:=Q\cup\{p_i\}$\;
} 
Update $I$\;
$\alpha:=\min\{s_i : i\in Q\}$\;}
\If{a packet $p_l$ is delivered}{
$Q:=Q-\{p_l\}$\;
 \If{ the queue is not empty}{
 Pick the packet with the largest generation time in the queue $p_h$\;
 Assign packet $p_h$ to an idle server\;
 $Q:=Q\cup\{p_h\}$\;
 }
 Update $I$\;
 $\alpha:=\min\{s_i : i\in Q\}$\;
}
}
\caption{Prmp-LGFS-R policy when $r=1$.}\label{alg1'}
\end{algorithm}

When there is a packet replication ($r>1$), the prmp-LGFS-R policy acts as follows. We replicate the packet with the largest generation time in the system on $r$ servers. Then, we replicate the packet with the second largest generation time in the system on the remaining idle servers such that the total number of replicas does not exceed $r$, and so on (i.e., the replicas of the packet with a larger generation time are sent with a higher priority than those of the packet with a lower generation time). In other words, since we may not have $m=ar$ for some positive integer $a$, packets under service may not be evenly distributed among the servers if all servers are busy. In this case, we give the highest priority to the $k$ ($k=\lfloor\frac{m}{r}\rfloor$) packets under service with the largest generation times and each one of them is replicated on $r$ servers. The packet under service with the smallest generation time is replicated on the remaining idle servers (whose number is less than $r$). If $m=ar$ for some positive integer $a$, then all packets under service are evenly distributed among the servers and each one of them is replicated on $r$ servers. The implementation details of prmp-LGFS-R policy when $r\geq 1$ are depicted in Algorithm \ref{alg1}: This algorithm explains the procedures that the prmp-LGFS-R policy follows in the case of packet arrival and departure events as follows.
\begin{itemize}
\item \textbf{Packet arrival event:} If a new packet $p_i$ arrives, we first check whether this new packet preempts an older packet that is being processed or not in Steps 6-19. After that, if packet $p_i$ is served, we specify the number of replicas that we need to create for packet $p_i$ in Steps 21-26. In particular, if packet $p_i$ is served, we have two possible cases.

\textbf{Case 1:} The generation time of packet $p_i$ is greater than the one with the smallest generation time in the set $Q$ (set $Q$ is defined at the beginning of the algorithm). In this case, we need to replicate packet $p_i$ on $r$ idle servers. Therefore, if the number of available servers ($I$) is less than $r$, we preempt $(r-I)$ more replicas of the packet with the smallest generation time in the set $Q$ and replicate packet $p_i$ on $r$ servers. These procedures are depicted in Steps 21-23.

\textbf{Case 2:} The generation time of packet $p_i$
is the smallest one among the packets in the set $Q$. In this case, packet $p_i$ is replicated on the available idle servers such that the total number of replicas of packet $p_i$ does not exceed $r$, as depicted in Steps 24-26.

\item \textbf{Packet departure event:} If a packet $p_l$ is delivered, we cancel all the remaining replicas of packet $p_l$. Moreover, if the queue is not empty, we pick the freshest packet in the queue and replicate it on the available idle servers such that the total number of replicas of this packet does not exceed $r$. These procedures are illustrated in Steps 29-39.
\end{itemize}
\begin{algorithm}[!t]
\footnotesize
\SetKwData{NULL}{NULL}
\SetCommentSty{footnotesize} 
$\alpha:=0$\tcp*[r]{$\alpha$ is the smallest generation time of the packets under service}
$I:=m$\tcp*[r]{$I$ is the number of idle servers}
$Q:=\emptyset$\tcp*[r]{$Q$ is the set of distinct packets that are under service}
$k:=\lfloor\frac{m}{r}\rfloor$\tcp*[r]{ $k$ is the number of distinct packets that each one of them can be replicated on $r$ servers}
\While{the system is ON} {
\If{a new packet $p_i$ with generation time $s$ arrives}{ 
\uIf(\tcp*[f]{All servers are busy}){$I=0$}{
\uIf(\tcp*[f]{Packet $p_i$ is stale}){$s\leq\alpha$}{
Store packet $p_i$ in the queue\;}
\Else(\tcp*[f]{Packet $p_i$ carries fresh information}){
Find packet $p_j\in Q$ with generation time $\alpha$\;
Preempt all replicas of packet $p_j$\;
Packet $p_j$ is stored back to the queue\;
$Q:=Q\cup\{p_i\} - \{p_j\}$\;
Update $I$\;
}}
\Else(\tcp*[f]{At least one of the servers is idle})
{
$Q:=Q\cup\{p_i\}$\;
} 
$\alpha:=\min\{s_i : i\in Q\}$\;
 
\uIf(\tcp*[f]{Specify the number of replicas of packet $p_i$}){$p_i\in Q$ \textbf{and} generation time of packet $p_i$ $>\alpha$ \textbf{and} $I<r$}{
 Preempt $(r-I)$ replicas of the packet with generation time $\alpha$\;
 Replicate packet $p_i$ on $r$ idle servers\;
}
\ElseIf{$p_i\in Q$ \textbf{and} generation time of packet $p_i$ $=\alpha$}{
  Replicate packet $p_i$ on $\min\{r,I\}$ idle servers\;
}
%
 Update $I$\;
}
\If{a packet $p_l$ is delivered}{
Cancel the remaining replicas of packet $p_l$\;
$Q:=Q-\{p_l\}$\;
 \If{ the queue is not empty}{
 Pick the packet with the largest generation time in the queue $p_h$\;
 $Q:=Q\cup\{p_h\}$\;
Replicate packet $p_h$ on $\min\{r,I\}$ idle servers\;
Update $I$\;
 }
$\alpha:=\min\{s_i : i\in Q\}$\;
}
}
\caption{Prmp-LGFS-R policy when $r\geq 1$.}\label{alg1}
\end{algorithm}

Note that the prmp-LGFS-R policy is a causal policy, i.e., its scheduling decisions are made based on the history and current state of the system and do not require the knowledge of the future packet generation and arrival times. Define a set of parameters $\mathcal{I}=\{B, m, r, s_i, a_{i},i=1,2,\ldots\}$, where $B$ is the queue buffer size, $m$ is the number of servers, $r$ is the maximum replication degree, $s_i$ is the generation time of packet $i$, and $a_{i}$ is the arrival time of packet $i$. Let $\Delta_{\pi}=\{\Delta_{\pi}(t), t\in [0,\infty)\}$ be the age processes under policy $\pi$. The age performance of the prmp-LGFS-R policy is characterized as follows.
\begin{theorem}\label{thm1}
Suppose that the packet service times are exponentially distributed, and \emph{i.i.d.} across servers and the packets assigned to the same server, then for all $\mathcal{I}$ and $\pi\in\Pi_{r}$
\begin{equation}\label{thm1eq1}
[\Delta_{\text{prmp-LGFS-R}}\vert\mathcal{I}]\leq_{\text{st}}[\Delta_{\pi}\vert\mathcal{I}],
\end{equation}
or equivalently, for all $\mathcal{I}$ and non-decreasing functional $g$
\begin{equation}\label{thm1eq2}
\mathbb{E}[g(\Delta_{\text{prmp-LGFS-R}})\vert\mathcal{I}]=\min_{\pi\in\Pi_{r}}\mathbb{E}[g(\Delta_\pi)\vert\mathcal{I}],
\end{equation}
provided the expectations in \eqref{thm1eq2} exist.
\end{theorem}

\begin{proof}
See Appendix~\ref{Appendix_A}.
\end{proof}

Theorem \ref{thm1} tells us that for arbitrary sequence of packet generation times $(s_1, s_2, \ldots)$, sequence of packet arrival times $(a_{1}, a_{2}, \ldots)$, buffer size $B$, number of servers $m$, and maximum replication degree $r$, the prmp-LGFS-R policy achieves optimality of the age process within the policy space $\Pi_{r}$. In addition, \eqref{thm1eq2} tells us that the prmp-LGFS-R policy minimizes any \emph{non-decreasing functional} of the age process, including the time-average age \eqref{functional1}, average peak age \eqref{functional2}, and time-average age penalty function \eqref{functional3} as special cases. It is important to emphasize that the prmp-LGFS-R policy can achieve optimality compared with all causal and non-causal policies in $\Pi_{r}$. Also, when the update packets arrive in the same order of their generation times and there is no replication, the prmp-LGFS-R policy becomes LCFS with preemption in service (LCFS-S) for a single source case that was proposed in \cite{RYatesTIT16}. Thus, this policy can achieve age-optimality in this case.

As a result of Theorem \ref{thm1}, we can deduce the following corollaries:

A weaker version of Theorem \ref{thm1} can be obtained as follows. 
\begin{corollary}\label{cor1''}
If the conditions of Theorem \ref{thm1} hold, then for any arbitrary packet generation and arrival processes, and for all $\pi\in\Pi_{r}$
\begin{equation*}
\Delta_{\text{prmp-LGFS-R}}\leq_{\text{st}}\Delta_{\pi}.
\end{equation*}
\end{corollary}
\begin{proof}
We consider the mixture over multiple sample paths of the packet generation and arrival processes to prove the result. In particular, by using the result of Theorem \ref{thm1} and Proposition \ref{Theorem_6.B.16.(e)}, the corollary follows. 
\end{proof}
\begin{corollary}\label{cor1'}
Under the conditions of Theorem \ref{thm1}, if one packet can be replicated to all $m$ servers (i.e., $r= m$), then for all $\mathcal{I}$, the prmp-LGFS-R policy when $r=m$ is an age-optimal among all policies in $\Pi_m$.
\end{corollary}
\begin{proof}
This corollary is a direct result of Theorem \ref{thm1}.
\end{proof}
It is important to recall that $\Pi_1\subset\Pi_2\subset\ldots\subset\Pi_m$. Therefore, Corollary \ref{cor1'} tells us that the prmp-LGFS-R policy when $r=m$ achieves age-optimality compared with all  policies with any maximum replication degree. 
\begin{corollary}\label{cor1}
If the conditions of Theorem \ref{thm1} hold, then for all $\mathcal{I}$, the age performance of the prmp-LGFS-R policy remains the same for any queue size $B\geq 0$.
\end{corollary}

\begin{proof}
From the operation of policy prmp-LGFS-R, its queue is used to store the preempted packets and outdated arrived packets. The age process of the prmp-LGFS-R policy is not affected no matter these packets are delivered or not. Hence, the age performance of  the prmp-LGFS-R policy is invariant for any queue size $B\geq 0$. This completes the proof.
\end{proof}
The next corollary clarifies the relationship between the prmp-LGFS-R policy and the LCFS-S policy.
\begin{corollary}
Under the conditions of Theorem \ref{thm1}, if the packets arrive to the queue in the same order of their generation times and replications are not allowed, then for all $\mathcal{I}$, the LCFS-S policy is age-optimal, i.e., the LCFS-S satisfies \eqref{thm1eq1} and \eqref{thm1eq2}.
\end{corollary}
\begin{proof}
This corollary is a direct result of Theorem \ref{thm1}.
\end{proof}
\subsubsection{Simulation Results}
We present some simulation results to compare the age performance of the prmp-LGFS-R policy with other policies. The packet service times are exponentially distributed with mean $1/\mu =1$. The inter-generation times are \emph{i.i.d.} Erlang-2 distribution with mean $1/\lambda$. The number of servers is $m$. Hence, the traffic intensity is $\rho=\lambda/m\mu$. \footnote{Throughout this paper, the traffic intensity $\rho$ is computed without considering replications (i.e., $\rho$ is calculated when $r=1$).} The queue size is $B$, which is a non-negative integer. 

\begin{figure}[h]
\centering
\includegraphics[scale=0.45]{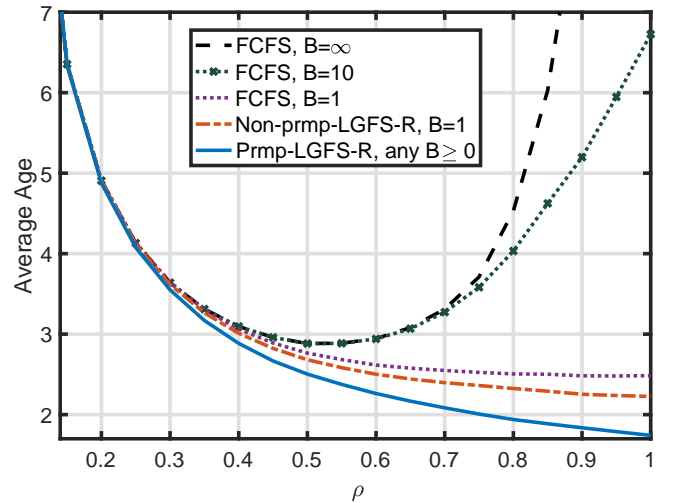}
\caption{Average age versus traffic intensity $\rho$ for an update system with $m=1$ server, queue size $B$, and  \emph{i.i.d.} exponential service times.}
\label{avg_age1}
\end{figure}
Figure \ref{avg_age1} illustrates the time-average age versus $\rho$ for an information-update system with $m=1$ server. The time difference ($a_i-s_i$) between packet generation and arrival is zero, i.e., the update packets arrive in the same order of their generation times. We can observe that the prmp-LGFS-R policy achieves a smaller age than the FCFS policy analyzed in \cite{KaulYatesGruteser-Infocom2012}, and the non-preemptive LGFS policy with queue size $B=1$ which is equivalent to ``M/M/1/2*'' in \cite{CostaCodreanuEphremides_TIT,CostaCodreanuEphremides2014ISIT} in this case. 
Note that in these prior studies, the age was characterized only for the special case of Poisson arrival process. Moreover, with ordered arrived packets at the server, the LGFS policy and LCFS policy have the same age performance.

\begin{figure}[h]
\centering
\includegraphics[scale=0.45]{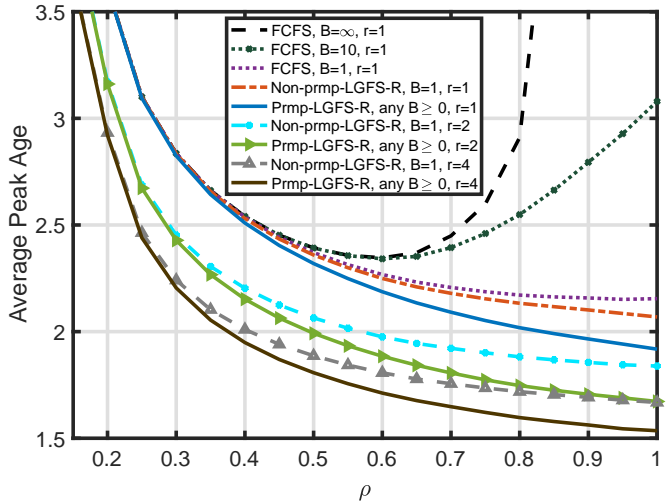}
\caption{Average peak age versus traffic intensity $\rho$ for an update system with $m=4$ servers, queue size $B$, maximum replication degree $r$, and \emph{i.i.d.} exponential service times.}
\label{avg_peak_age2}
\end{figure}
Figure \ref{avg_peak_age2} plots the average peak age versus $\rho$ for an information-update system with $m=4$ servers. The time difference between packet generation and arrival, i.e., $a_i-s_i$, is modeled to be either 1 or 100, with equal probability. The maximum replication degree $r$ is either 1, 2, or 4. For each $r$, we found that the prmp-LGFS-R policy achieves better age performance than other policies that belong to the policy space $\Pi_r$. For example, the age performance of the prmp-LGFS-R policy when $r=2$ is better than the age performance of the other policies that are plotted when $r$ equal to 1 and 2. Note that the age performance of the prmp-LGFS-R policy remains the same for any queue size $B\geq0$. However, the age performance of the non-prmp-LGFS-R policy and FCFS policy varies with the queue size $B$. We also observe that the average peak age in case of FCFS policy with $B=\infty$ blows up when the traffic intensity is high. This is due to the increased congestion in the network which leads to a delivery of stale packets. Moreover, in case of FCFS policy with $B=10$, the average peak age is high but bounded at high traffic intensity, since the fresh packet has a better opportunity to be delivered in a relatively short period compared with FCFS policy with $B=\infty$. These numerical results agree with Theorem \ref{thm1}.

\subsection{NBU Service Time Distributions}
The next question we proceed to answer is whether for an important class of distributions that are more general than exponential, age-optimality or near age-optimality can be achieved. We consider the class of NBU packet service time distributions, which are defined as follows.
\begin{definition}  \textbf{New-Better-than-Used distributions:} Consider a non-negative random variable $Z$ with complementary cumulative distribution function (CCDF) $\bar{F}(z)=\mathbb{P}[Z>z]$. Then, $Z$ is \textbf{New-Better-than-Used (NBU)} if for all $t,\tau\geq0$
\begin{equation}\label{NBU_Inequality}
\bar{F}(\tau +t)\leq \bar{F}(\tau)\bar{F}(t).
\end{equation} 
\end{definition}
Examples of NBU distributions include constant service time, Gamma distribution, (shifted) exponential distribution, geometric distribution, Erlang distribution, negative binomial distribution, etc. 

\begin{algorithm}[!t]
\footnotesize
\SetKwData{NULL}{NULL}
\SetCommentSty{footnotesize} 
$\delta:=0$\tcp*[r]{$\delta$ is the smallest generation time of the packets in the queue}
$I:=m$\tcp*[r]{$I$ is the number of idle servers}
$k:=\lfloor\frac{m}{r}\rfloor$\tcp*[r]{$k$ is number of packets that each one of them can be replicated on $r$ servers}
\While{the system is ON} {
\If{a new packet $p_i$ with generation time $s$ arrives}{ 
\uIf(\tcp*[f]{All servers are busy}){I=0}{
\uIf{Buffer is full}{
\uIf(\tcp*[f]{Packet $p_i$ carries fresh information}){$s>\delta$}{
 Packet $p_i$ replaces the packet with generation time $\delta$ in the queue\;
 }
 \Else (\tcp*[f]{Packet $p_i$ is stale}){
  Drop packet $p_i$\;}}
  \Else{
Store packet $p_i$ in the queue\;
}Update $\delta$\;}
\Else(\tcp*[f]{At least one of the servers is idle})
{
Replicate packet $p_i$ on $\min\{r,I\}$ idle servers\;
Update $I$\; 
}
\If{a packet $p_l$ is delivered}{
Cancel the remaining replicas of packet $p_l$\;
Update $I$\;
Find packet $p_j$ that is replicated on $(m-kr)$ servers\;
\uIf{the queue is empty \textbf{and} packet $p_j$ exists}{
Replicate packet $p_j$ on extra $((k+1)r-m)$ idle servers\;}
\ElseIf{the queue is not empty}{
Pick the packet with the largest generation time in the queue $p_h$\;
\If{packet $p_j$ exists \textbf{and} generation time of packet $p_j>$ generation time of packet $p_h$}{
Replicate packet $p_j$ on extra $((k+1)r-m)$ idle servers\;
Update $I$\;}
Replicate packet $p_h$ on $\min\{r,I\}$ idle servers\;
}
Update $I$\;
Update $\delta$\;}
}}
\caption{Non-prmp-LGFS-R policy.}\label{alg2}
\end{algorithm}

Next, we show that near age-optimality can be achieved when the service times are NBU. We propose a policy called non-preemptive LGFS with replication (non-prmp-LGFS-R). The non-prmp-LGFS-R policy has the same main features of the prmp-LGFS-R policy except that the non-prmp-LGFS-R policy does not allow packet preemption. Moreover, under the non-prmp-LGFS-R policy, the fresh packet replaces the packet with the smallest generation time in the queue when it has a finite buffer size that is full. The description of the non-prmp-LGFS-R policy is depicted in Algorithm \ref{alg2}: This algorithm explains the procedures that the non-prmp-LGFS-R policy follows in the case of packet arrival and departure events as follows.
\begin{itemize}
\item \textbf{Packet arrival event:} If a new packet $p_i$ arrives and all servers are busy, then we have two cases.

\textbf{Case 1:} The buffer is full. In this case, packet $p_i$ is either dropped or replaces another packet in the queue depending on its generation time, as depicted in Steps 7-12.

\textbf{Case 2:} The buffer is not full. In this case, packet $p_i$ is stored directly in the queue, as depicted in Steps 13-15. 

If there are idle servers, then packet $p_i$ is replicated on the available idle servers such that the total number of replicas of packet $p_i$ does not exceed $r$, as illustrated in Steps 17-20.

\item \textbf{Packet departure event:} If a packet $p_l$ is delivered, we cancel all the remaining replicas of packet $p_l$. Also, if there is a packet $p_j$ that is replicated on fewer servers than $r$ servers, then packet $p_j$ is replicated on extra $((k+1)r-m)$ servers under two cases.

 \textbf{Case A:} If the queue is empty, as depicted in Steps 24-26.
 
 \textbf{Case B:} If the queue is not empty, but the generation time of packet $p_j$ is greater than the largest generation time of the packets in the queue, as depicted in Steps 27-32. 
 
 Finally, if the queue is not empty, the packet with the largest generation time in the queue is replicated on the available idle servers such that the total number of replicas of this packet does not exceed $r$, as illustrated in Step 33.
\end{itemize}

It is important to emphasize that the non-prmp-LGFS-R policy is a causal policy, i.e., its scheduling decisions are made based on the history and current state of the system and do not require the knowledge of the future packet generation and arrival times. To show that policy non-prmp-LGFS-R can come close to age-optimal, we need to construct an age lower bound as follows: 

 Let $v_i$ denote the earliest time that packet $i$ has started service (the earliest assignment time of packet $i$ to a server), which is a function of the scheduling policy $\pi$.
\begin{figure}
\includegraphics[scale=0.4]{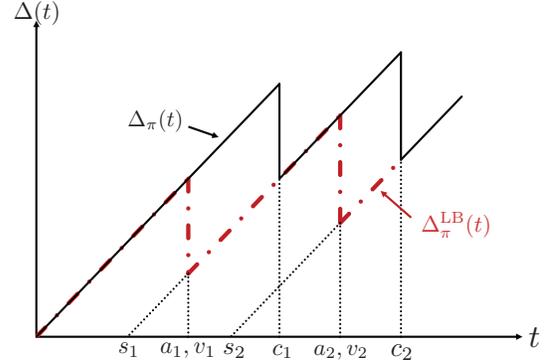}
\centering
\caption{The evolution of $\Delta_{\pi}^{\text{LB}}$ and $\Delta_{\pi}$ in a single server queue. We assume that $a_1>s_1$ and $a_2>c_1>s_2$. Thus, we have $v_1=a_1$ and $v_2=a_2$.}\label{Fig:lower_bound_shap}
\vspace{-0.3cm}
\end{figure} 
Define a function $\Delta_\pi^{\text{LB}}(t)$ as 
\begin{equation}\label{lowerbound}
\Delta_\pi^{\text{LB}}(t)=t-\max\{s_i : v_i(\pi)\leq t\}.
\end{equation} 
The process of $\Delta_\pi^{\text{LB}}(t)$ is given by $\Delta_\pi^{\text{LB}}=\{\Delta_\pi^{\text{LB}}(t), t\in [0,\infty)\}$. The definition of the process $\Delta_\pi^{\text{LB}}(t)$ is similar to that of the age process of policy $\pi$ except that the packet completion times are replaced by their assignment times to the servers. In this case, the process $\Delta_\pi^{\text{LB}}(t)$ increases linearly with $t$ but is reset to a smaller value with the assignment of a fresher packet to a server under policy $\pi$, as shown in Fig. \ref{Fig:lower_bound_shap}.
The process $\Delta_{\text{non-prmp-LGFS-R}}^{\text{LB}}$ is a lower bound of all policies in $\Pi_m$ in the following sense.
\begin{lemma}\label{lowerbound_lemma}
Suppose that the packet service times are NBU, and \emph{i.i.d.} across servers and the packets assigned to the same server, then for all $\mathcal{I}$ satisfying $B\geq 1$, and $\pi\in\Pi_{m}$
\begin{equation}\label{hypo_na2}
[\Delta_{\text{non-prmp-LGFS-R}}^{\text{LB}}\vert\mathcal{I}]\leq_{\text{st}}[\Delta_{\pi}\vert\mathcal{I}].
\end{equation}
\end{lemma}
\begin{proof}
See Appendix~\ref{Appendix_B}.
\end{proof}
We can now proceed to characterize the age performance of policy non-prmp-LGFS-R. 
Let $X_1,\ldots, X_m$ denote the \emph{i.i.d.} packet service times of the $m$ servers, with mean $E[X_l] = E[X] <\infty$. We use Lemma \ref{lowerbound_lemma} to prove the following theorem. 
\begin{theorem}\label{thmNBU1}
Suppose that the packet service times are NBU, and \emph{i.i.d.} across servers and the packets assigned to the same server, then for all $\mathcal{I}$ satisfying $B\geq 1$
\vspace{.25cm}\begin{itemize}
\abovedisplayskip=-\baselineskip
\abovedisplayshortskip=-\baselineskip
\item[(a)] 
\begin{equation}\label{NBUhypoeq1}
\begin{split}
\min_{\pi\in\Pi_m} [\bar{\Delta}_{\pi}\vert\mathcal{I}]\leq [\bar{\Delta}_{\text{non-prmp-LGFS-R}}\vert\mathcal{I}]\leq \\\min_{\pi\in\Pi_m} [\bar{\Delta}_{\pi}\vert\mathcal{I}]+\mathbb{E}[X].
\end{split}
\end{equation}
\end{itemize}
\begin{itemize}
\item[(b)] If there is a positive integer $a$ such that $m = ar$, then
\begin{equation}\label{NBUhypoeq2}
\begin{split}
\min_{\pi\in\Pi_m} [\bar{\Delta}_{\pi}\vert\mathcal{I}]\leq [\bar{\Delta}_{\text{non-prmp-LGFS-R}}\vert\mathcal{I}]\leq \\\min_{\pi\in\Pi_m} [\bar{\Delta}_{\pi}\vert\mathcal{I}]+\mathbb{E}[\min_{l=1,\ldots,r}X_{l}],
\end{split}
\end{equation} 
\end{itemize}
where $\bar{\Delta}_{\pi}=\limsup_{T\rightarrow\infty}\frac{\mathbb{E}[\int_0^{T} \Delta_{\pi}(t)dt]}{T}$ is the average age under policy $\pi$.
\end{theorem}

\begin{proof}
See Appendix~\ref{Appendix_D}.
\end{proof}

Theorem \ref{thmNBU1} tells us that for arbitrary sequence of packet generation times $(s_1, s_2, \ldots)$, sequence of packet arrival times $(a_{1}, a_{2}, \ldots)$, number of servers $m$, maximum replication degree $r$, and buffer size $B\geq1$, the non-prmp-LGFS-R policy is within a constant age gap from the optimum average age among policies in $\Pi_m$. It is important to emphasize that policy non-prmp-LGFS-R with a maximum replication degree $r$ can be near age-optimal compared with policies with any maximum replication degree. Also, when the update packets arrive in the same order of their generation times and there is no replication, the non-prmp-LGFS-R policy when $B=1$ becomes LCFS with preemption only in waiting (LCFS-W) for a single source case in \cite{RYatesTIT16}, or the ``M/M/1/2*'' in \cite{CostaCodreanuEphremides_TIT,CostaCodreanuEphremides2014ISIT}. Thus, these policies can achieve near age-optimality in this case. The following corollary emphasizes this relationship.
\begin{corollary}
Under the conditions of Theorem \ref{thmNBU1}, if the packets arrive to the queue in the same order of their generation times, replications are not allowed ($r=1$), and $B=1$, then for all $\mathcal{I}$, the LCFS-W policy and the ``M/M/1/2*'' policy  are near age-optimal, i.e., these policies satisfy \eqref{NBUhypoeq1}.
\end{corollary}

\begin{proof}
This corollary is a direct result of Theorem \ref{thmNBU1}.
\end{proof}

\subsubsection{Simulation Results}
We now provide simulation results to illustrate the age performance of different policies when the service times are NBU. The inter-generation times are \emph{i.i.d.} Erlang-2 distribution with mean $1/\lambda$. The time difference ($a_i-s_i$) between packet generation and arrival  is zero. The maximum replication degree $r$ is either 1 or 4.

\begin{figure}[h]
\centering
\includegraphics[scale=0.45]{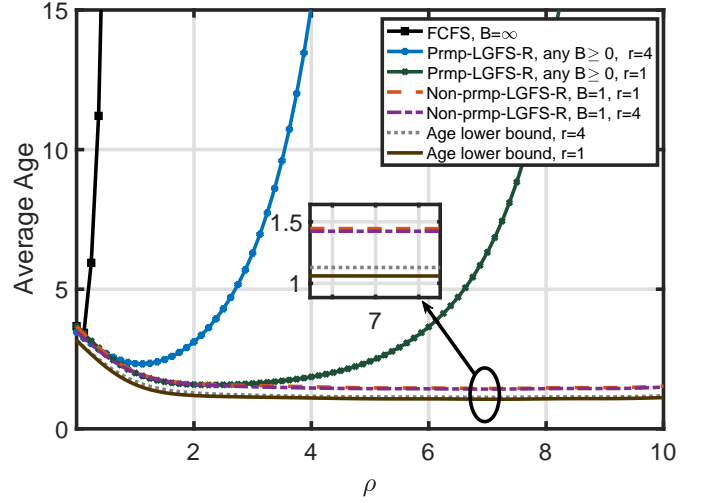}
\caption{Average age versus traffic intensity $\rho$ for an update system with $m=4$ servers, queue size $B$, maximum replication degree $r$, and \textit{i.i.d} NBU service times.}
\label{avg_age4}
\end{figure}
Figure \ref{avg_age4} plots the average age versus $\rho$ for an information-update system with $m=4$ servers. The packet service times are the sum of a constant .25 and a value drawn from an exponential distribution with mean .25. Hence, the mean service time is $1/\mu=.5$. The ``Age lower bound'' curves are generated by using $\frac{\int_0^{T}\Delta_{\text{non-prmp-LGFS-R}}^{\text{LB}}(t)dt}{T}$ when $r$ is 1 and 4, and $B=1$ which, according to Lemma \ref{lowerbound_lemma}, are lower bounds of the optimum average age. We can observe that the gap between the ``Age lower bound'' curves and the average age of the non-prmp-LGFS-R policy when $r=1$ and $r=4$ is no larger than $E[X] = 1/\mu = .5$, which agrees with Theorem \ref{thmNBU1}. This is a surprising result since it was shown in \cite{sun2016delay, shah2016redundant, joshi2017efficient} that replication policies are far from the optimum delay and throughput performance for NBU service time distributions. Moreover, we can observe that the average age of the prmp-LGFS-R policies blows up when the traffic intensity is high.
This is because the packet service times do not have the memoryless property in this case. Hence, when a packet is preempted, the service time of a new packet is probably longer than the remaining service time of the preempted packet. Because the arrival rate is high, packet preemption happens frequently, which leads to infrequent packet delivery and increases the age, as observed in \cite{CostaCodreanuEphremides_TIT}.

\begin{figure}[h]
\centering
\includegraphics[scale=0.45]{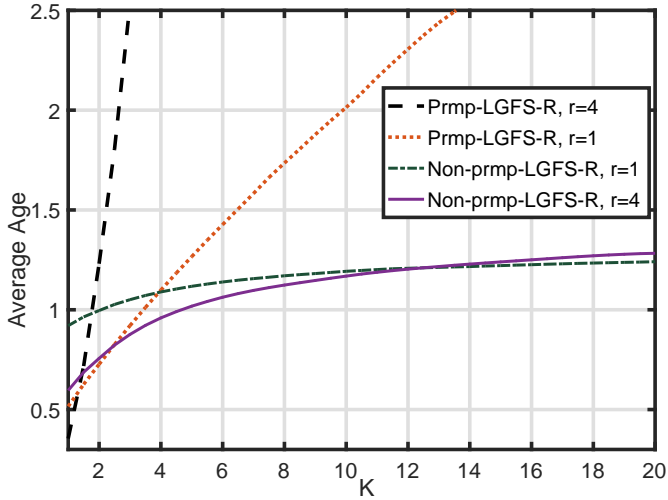}
\caption{Average age under gamma service time distributions with different shape parameter $K$, where $m=4$ servers, queue size $B=\infty$, and maximum replication degree $r$.}
\label{avg_peak_age3}
\end{figure}
Figure \ref{avg_peak_age3} plots the average age under gamma service time distributions with different shape parameter $K$, where $m=4$, $B=\infty$, and the traffic intensity $\rho=\lambda/m\mu=1.8$. The mean of the gamma service time distributions is normalized to $1/\mu =1$. Note that the average age of the FCFS policy in this case is extremely high and hence is not plotted in this figure. One can notice that packet replication and preemption affect the age performance of the plotted policies. In particular, we found that packet replication improves the age performance of the non-prmp-LGFS-R policy when the shape parameter $K\leq 12.5$, where the non-prmp-LGFS-R policy for $r=4$ outperforms the case of $r=1$. This is because the variance (variability) of the normalized gamma distribution is high for small values of $K$. Thus, packet replication can exploit the diversity provided by the four servers in this case. For the same reason, we can observe that packet replication improves the age performance of the preemptive policies when $K=1$, where the prmp-LGFS-R policy for $r=4$ achieves the best age performance among all plotted policies. Another reason behind the latter observation is that a gamma distribution with shape parameter $K=1$ is an exponential distribution and hence is memoryless. Thus, packet preemption improves the age performance in this case and age-optimality can be achieved by the prmp-LGFS-R policy when $r=m$ as stated in Theorem \ref{thm1} and Corollary \ref{cor1'}. On the other hand, as the shape parameter $K$ increases, the variance (variability) of the normalized gamma distribution decreases. This, in turn, reduces the benefit gained from the diversity provided by four servers and hence worsens the age performance of the policies that use packet replication. Moreover, as can be seen in the figure, preemption further worsens the age performance as the shape parameter $K$ increases, and the average age of the prmp-LGFS-R policies blows up in this case. This is because of the reduction in the variability of the packet service time when the shape parameter $K$ increases as well as the loss of the memoryless property when $K\neq 1$. Thus, preemption is not useful in this case.

\subsection{Discussion}\label{Disc}
In this subsection, we discuss our results and compare it with prior works. 
\subsubsection{Preemption vs. Non-Preemption}
The effect of the preemption on the age performance depends basically on the distribution of the packet service time. More specifically, when the packet service times are exponentially distributed, preemptive policies (i.e., prmp-LGFS-R) can achieve age-optimality (Theorem \ref{thm1}). This is because the remaining service time of a preempted packet has the same distribution as the service time of a new packet. For example, in Fig. \ref{avg_peak_age3}, preemptive policies provide the best age performance when $K=1$ (gamma distribution with shape parameter $K=1$ is an exponential distribution). It is important to notice that preemptive policies can achieve age-optimality regardless of the value of $\rho$, even if the system is unstable when $\rho>1$ ($\rho=1.8$ in Fig. \ref{avg_peak_age3}). Thus, we suggest using preemption when the packet service times are exponentially distributed. However, when the packet service times are NBU, we suggest to not use preemption. This is because the service times are no longer memoryless. Hence, when a packet is preempted, the service time of a new packet is probably longer than the remaining service time of the preempted packet. As shown in Fig. \ref{avg_peak_age3}, the age of the preemptive LGFS policy grows to infinity at high traffic intensity for gamma distributed service times with $K>1$. Thus, we suggest using non-preemptive policies (i.e., non-prmp-LGFS-R) instead when the packet service times are NBU. 

Similar observations have been made in previous studies \cite{RYatesTIT16,Gamma_dist}. For exponential service time distribution, Yates and Kaul showed in Theorem 3(a) of \cite{RYatesTIT16} that the average age of the preemptive LCFS policy is a decreasing function of the traffic intensity $\rho$ in M/M/1 queues as $\rho$ grows to infinite. This agrees with our study, in which we proved that the preemptive LCFS policy is age-optimal for exponential service times and general system parameters. For NBU service time distributions, our study agrees with \cite{Gamma_dist}. In particular, in \cite[Numerical Results]{Gamma_dist}, the authors showed that the non-preemptive LCFS policy can achieve better average age than the preemptive LCFS policy. In this paper, we further show that the non-prmp-LGFS-R policy is within a small constant gap from the optimum age performance for all NBU service time distributions, which include gamma distribution as one example.

In general, our study was carried out for system settings that are more general than \cite{RYatesTIT16} and \cite{Gamma_dist}.

\subsubsection{Replication vs. Non-Replication}
The replication technique has gained significant attention in recent years to reduce the delay in queueing systems \cite{Chen_queue_coding,yin_data_retrieve,sun2016delay}. However, it was shown in \cite{sun2016delay, shah2016redundant, joshi2017efficient} that replication policies are far from the optimum delay and throughput performance for NBU service time distributions. A simple explanation of this result is as follows: Let $X_1, \ldots, X_m$ be \emph{i.i.d}. NBU random variables with mean $\mathbb{E}[X_l]=\mathbb{E}[X]<\infty$. From the properties of the NBU distributions, we can obtain \cite{shaked2007stochastic}
\begin{equation}\label{NBUproperty_delay}
\frac{1}{\mathbb{E}[\min_{l=1,\ldots,m}X_l]}\leq \frac{m}{\mathbb{E}[X]}.
\end{equation}
Now, if $X_l$ represents the packet service time of server $l$, then the left-hand side of \eqref{NBUproperty_delay} represents the service rate when each packet is replicated to all servers; and the right-hand side of \eqref{NBUproperty_delay} represents the service rate when there is no replication. This gives insight why packet replication can worsen the delay and throughput performance when the service times are NBU. 

Somewhat to our surprise, we found that the non-prmp-LGFS-R policy is near-optimal in minimizing the age, even for NBU service time distributions. The intuition behind this result is that the age is affected by only the freshest packet, instead of all the packets in the queue. In other words, to reduce the age, we need to deliver the freshest packet as soon as possible. Obviously, we have
\begin{equation}
\mathbb{E}[\min_{l=1,\ldots,m}X_l]\leq\mathbb{E}[X].
\end{equation}
Thus, packet replication can help to reduce the age by exploiting the diversity provided by multiple servers. As shown in Fig. \ref{avg_peak_age3}, we can observe that packet replication can improve the age performance. In particular, the age performance of the non-prmp-LGFS-R policy with $r=4$ is better than that of the non-prmp-LGFS-R policy with $r=1$ when $K\leq 12.5$.

\section{Throughput-Delay Analysis}\label{thro-anal}
Recent studies on information-update systems have shown that the age-of-information can be reduced by intelligently dropping stale packets. However, packet dropping may not be appropriate in many applications such as (but not limited to):
\begin{itemize}
\item \textbf{News feeds:} In addition to the latest breaking news, the older news may be relevant to the user as well (e.g., to provide context or outline a different story that the user may have missed, etc.). 
\item \textbf{Social updates:} Users may need to be up to date with the freshest events and social posts. Nonetheless, they may also be interested in the previous posts. Thus, social applications need to update users with latest posts and previous ones as well. 
\item \textbf{Stock quotes:} Although the latest price in the market is
very important for the traders, they may also use the history of the price change to predict the short-term price movement and attempt to profit from this. Thus, both the latest prices and historical price data  are important in this case.
\item \textbf{Autonomous driving or sensor information:} In such applications, while it is important to receive the latest information, historical information may also be relevant to exploit trends. For example, historical data on location information can predict the trajectory, velocity, and acceleration of the automobile. Similarly, certain types of historical sensed data may be useful to predict forest fires, earthquakes, Tsunamis, etc.
\end{itemize}
 In these applications, users are interested in not just the latest updates, but also past information. Therefore, all packets may need to be successfully delivered. This motivates us to study whether it is possible to simultaneously optimize multiple performance metrics, such as age, throughput, and delay. In the sequel, we investigate the throughput and delay performance of the proposed policies. We first consider the exponential service time distribution. Then, we generalize the service time distribution to the NBU distributions. We need the following definitions:
\begin{definition}\textbf{Throughput-optimality:} A policy is said to be \textbf{throughput-optimal}, if it maximizes the expected number of delivered packets among all policies. 
\end{definition}
The average delay under policy $\pi$ is defined as
\begin{equation}\label{average_delay}
D_{\text{avg}}(\pi)=\frac{1}{n}\sum_{i=1}^n[c_i(\pi)-a_i],
\end{equation}
where the delay of packet $i$ under policy $\pi$ is $c_i(\pi)-a_i$.\footnote{The $\limsup$ operator is enforced on the right hand side of \eqref{average_delay} if $n\rightarrow\infty$.}
\begin{definition}\textbf{Delay-optimality:} A policy is said to be \textbf{delay-optimal}, if it minimizes the expected average delay among all policies. 
\end{definition}
Note that to maximize the throughput, we need to maximize the total number of distinct delivered packets. Moreover, to minimize the expected average delay, we need to minimize the total number of distinct packets in the system along the time. Based on these two key ideas, we prove our results in the next subsections. 
\subsection{Exponential Service Time Distribution}
We study the throughput and delay performance of the prmp-LGFS-R policy when the service times are \emph{i.i.d.} exponentially distributed. The delay and throughput performance of the prmp-LGFS-R policy are characterized as follows:
\begin{theorem}\label{thm2th}
Suppose that the packet service times are exponentially distributed, and \emph{i.i.d.} across servers and the packets assigned to the same server, then for all $\mathcal{I}$ such that $B=\infty$, the prmp-LGFS-R policy is throughput-optimal and delay-optimal among all policies in $\Pi_m$.
\end{theorem}

\begin{proof}
We provide a proof sketch of Theorem \ref{thm2th}. We use the coupling and forward induction to prove it. We first consider the comparison between the prmp-LGFS-R policy and an arbitrary work-conserving policy $\pi$. We couple the packet departure processes at each server such that they are identical under both policies. Then, we use the forward induction over the packet arrival and departure events to show that the total number of distinct packets in the system (excluding packet replicas) and the total number of distinct delivered packets are the same under both policies. By this, we show that the prmp-LGFS-R policy has the same throughput and mean-delay performance as any work-conserving policy (indeed, all work-conserving policies have the same throughput and delay performance). Finally, since the packet service times are  \emph{i.i.d.} across servers and the packets assigned to the same server, service idling only postpones the delivery of packets. Therefore, both throughput and delay under non-work-conserving policies will be worse. For more details, see Appendix~\ref{Appendix_C}.
\end{proof}
It is worth pointing out that when the packet service times are \emph{i.i.d} exponentially distributed, packet replication does not affect the throughput and delay performance of the replicative policies. The reasons for this observation can be summarized as follows. Because the packet service times are \emph{i.i.d.} across the servers and the CCDF $\bar{F}$ is continuous, the  probability  for  any  two  servers  to  complete  their  packets  at  the  same  time  is  zero.  Therefore,  in the replicative policies, if  one  copy  of  a  replicated  packet  is  completed  on  a  server,  the  remaining  replicated  copies of this packet are still being processed on the other servers; these replicated packet copies are cancelled immediately and a new packet is replicated on these servers. Due to the memoryless property of the exponential distribution, the service times of the new packet copies and the remaining service times of the cancelled packets have the same distribution. Thus, packet replication does not affect the throughput and delay performance of the replicative policies.

\subsection{NBU Service Time Distributions}
Now, we consider a class of NBU service time distributions. We study the throughput and delay performance of the non-prmp-LGFS-R policy when there is no replication.
The delay and throughput performance of the non-prmp-LGFS-R policy are characterized as follows:
\begin{theorem}\label{thmNBU1th}
Suppose that the packet service times are NBU, and \emph{i.i.d.} across servers and the packets assigned to the same server, then for all $\mathcal{I}$ such that $B=\infty$ and $r=1$, the non-prmp-LGFS-R policy is throughput-optimal and delay-optimal among all non-preemptive policies in $\Pi_{1}$.
\end{theorem}

We omit the proof of Theorem \ref{thmNBU1th}, because it is similar to that of Theorem \ref{thm2th}.



\section{Conclusions}\label{Concl}
In this paper, we studied the age-of-information optimization in multi-server queues. Packet replication was considered in this model, where the maximum replication degree is constrained. We considered general system settings including arbitrary arrival processes where the incoming update packets may arrive \emph{out of order} of their generation times. We developed scheduling policies that can achieve age-optimality for any maximum replication degree when the packet service times are exponentially distributed. This optimality result not only holds for the age process, but also for any \emph{non-decreasing functional} of the age process. Interestingly, the proposed policies can also achieve throughput and delay optimality. In addition, we investigated the class of NBU packet service time distributions and showed that LGFS policies with replication are near age-optimal for any maximum replication degree. 



%
%
%
\appendices
\section{Proof of Theorem \ref{thm1}}\label{Appendix_A}
We need to define the system state of any policy $\pi$:
\begin{definition}  Define $U_{\pi}(t)$ as the largest generation time of the packets at the destination at time $t$ under policy $\pi$. Let $\alpha_{i,\pi}(t)$ be the generation time of the packet that is being processed by server $i$ at time $t$ under policy $\pi$, where we set $\alpha_{i,\pi}(t)=U_{\pi}(t)$ if server $i$ is idle. Then, at any time $t$, the system state of policy $\pi$ is specified by  $\mathbf{V}_{\pi}(t)=(U_{\pi}(t), \alpha_{[1],\pi}(t),$ $\ldots,\alpha_{[m],\pi}(t)) $. Note that if there is a replication, we may have $ \alpha_{[i],\pi}(t)= \alpha_{[i+1],\pi}(t)$ for some $i$'s. Without loss of generality, if $h$ servers are sending packets with generation times less than $U_\pi(t)$ (i.e., $\alpha_{[m],\pi}(t) \leq \alpha_{[m-1],\pi}(t)\leq \ldots\leq \alpha_{[m-h+1],\pi}(t) \leq U_\pi(t)$) or $h$ servers are idle, then we set $\alpha_{[m],\pi}(t) =\ldots =\alpha_{[m-h+1],\pi}(t) = U_\pi(t)$. Hence, 
\begin{equation}\label{sys_state_cond1}
U_\pi (t)\leq \alpha_{[m],\pi}(t)\leq\ldots\leq\alpha_{[1],\pi}(t). 
\end{equation}
Let $\{\mathbf{V}_{\pi}(t), t\in[0,\infty)\}$ be the state process of policy $\pi$, which is assumed to be right-continuous. For notational simplicity, let policy $P$ represent the prmp-LGFS-R policy. Throughout the proof, we assume that  $\mathbf{V}_P(0^-)=\mathbf{V}_{\pi}(0^-)$ for all $\pi\in\Pi_r$.
\end{definition}

The key step in the proof of Theorem \ref{thm1} is the following lemma, where we compare policy $P$ with any work-conserving policy $\pi$.

 \begin{lemma}\label{lem2}
Suppose that $\mathbf{V}_P(0^-)=\mathbf{V}_{\pi}(0^-)$ for all work conserving policies $\pi$, then for all $\mathcal{I}$
\begin{equation}\label{law9}
\begin{split}
 [\{\mathbf{V}_P(t),  t\in[0,\infty)\}\vert\mathcal{I}]\!\geq_{\text{st}}\! [\{\mathbf{V}_{\pi}(t), t\in[0,\infty)\}\vert\mathcal{I}].
 \end{split}
\end{equation}
\end{lemma}

 We use coupling and forward induction to prove Lemma \ref{lem2}.
For any work-conserving policy $\pi$, suppose that stochastic processes $\widetilde{\mathbf{V}}_{P}(t)$ and $\widetilde{\mathbf{V}}_{\pi}(t)$ have the same stochastic laws as $\mathbf{V}_{P}(t)$  and $\mathbf{V}_{\pi}(t)$. 
The state processes $\widetilde{\mathbf{V}}_{P}(t)$ and $\widetilde{\mathbf{V}}_{\pi}(t)$
are coupled in the following manner: If the packet with generation time $\widetilde{\alpha}_{[i],P}(t)$ is delivered at time $t$ as $\widetilde{\mathbf{V}}_{P}(t)$ evolves, then the packet with generation time $\widetilde{\alpha}_{[i],\pi}(t)$ is delivered at time $t$  as $\widetilde{\mathbf{V}}_{\pi}(t)$ evolves. 
Such a coupling is valid because the service times are exponentially distributed and thus memoryless. Moreover, policy $P$ and policy $\pi$ have identical packet generation times $(s_1, s_2, \ldots)$ and packet arrival times $(a_1, a_2, \ldots)$. According to Proposition \ref{Theorem_6.B.30}, if we can show 
\begin{equation}\label{main_eq}
\begin{split}
\mathbb{P}[\widetilde{\mathbf{V}}_{P}(t)\geq\widetilde{\mathbf{V}}_{\pi}(t), t\in[0,\infty)\vert\mathcal{I}]=1,
\end{split}
\end{equation}
then \eqref{law9} is proven.
To ease the notational burden, we will omit the tildes on the coupled versions in this proof and just use $\mathbf{V}_P(t)$ and $\mathbf{V}_{\pi}(t)$. Next, we use the following lemmas to prove \eqref{main_eq}:
\begin{lemma}\label{lem3'}
At any time $t$, suppose that the system state of policy $P$ is $\{U_P,  \alpha_{[1],P},\ldots, \alpha_{[m],P}\}$, and meanwhile the system state of policy $\pi$ is $\{U_{\pi}, \alpha_{[1],\pi},\ldots,\alpha_{[m],\pi}\}$. If
\begin{equation}\label{hyp1'}
U_P \geq U_{\pi},
\end{equation}
then,
\begin{equation}\label{law6'}
 \alpha_{[i],P} \geq \alpha_{[i],\pi}, \quad \forall i=1,\ldots,m.
\end{equation}
\end{lemma}

\begin{proof}
Let $S$ denote the set of packets that have arrived to the system at the considered time $t$. It is important to note that the set $S$ is invariant of the scheduling policy. If $S$ is empty, then since $\mathbf{V}_P(0^-)=\mathbf{V}_{\pi}(0^-)$, Lemma \ref{lem3'} follows directly. Thus, we assume that $S$ is not empty during the proof. We use $s_{[i]}$ to denote the $i$-th largest generation time of the packets in $S$. Define $k=\lfloor\frac{m}{r}\rfloor$. From the definition of the system state, condition \eqref{sys_state_cond1}, and the definition of policy $P$, we have
\begin{equation}\label{lnew1}
\begin{split}
 \alpha_{[i],P}=&\max\{s_{[j]},U_P\},\\&\forall i=(j-1)r+1,\ldots,jr,~\forall j=1,\ldots,k,\\
 \alpha_{[i],P}=&\max\{s_{[k+1]},U_P\},~\forall i=kr+1,\ldots,m.
 \end{split}
\end{equation}
Since policy $\pi$ is an arbitrary policy, the servers under policy $\pi$ may not process the packets with the largest generation times in the set $S$ or policy $\pi$ may replicate packets with lower generation times more than those that have larger generation times in the set $S$. Hence, we have
\begin{equation}\label{lnew2}
\begin{split}
 \alpha_{[i],\pi}\leq &\max\{s_{[j]},U_\pi\},\\&\forall i=(j-1)r+1,\ldots,jr,~\forall j=1,\ldots,k,\\
 \alpha_{[i],\pi}\leq &\max\{s_{[k+1]},U_\pi\},~\forall i=kr+1,\ldots,m.
 \end{split}
\end{equation}
where the maximization here follows from the definition of the system state. Since the set $S$ is invariant of the scheduling policy and $U_P \geq U_{\pi}$, this with \eqref{lnew1} and \eqref{lnew2} imply
\begin{equation}
 \alpha_{[i],P} \geq \alpha_{[i],\pi}, \quad \forall i=1,\ldots,m,
\end{equation}
which completes the proof.
\end{proof}

\begin{lemma}\label{lem3}
Suppose that under policy $P$, $\{U_P',  \alpha_{[1],P}',\ldots, \alpha_{[m],P}'\}$ is obtained by delivering a packet with generation time $\alpha_{[l],P}$ to the destination in the system whose state is $\{U_P,  \alpha_{[1],P},\ldots, \alpha_{[m],P}\}$. Further, suppose that under policy $\pi$, $\{U_{\pi}', \alpha_{[1],\pi}',\ldots,\alpha_{[m],\pi}'\}$ is obtained  by delivering a packet with generation time $\alpha_{[l],\pi}$ to the destination in the system whose state is $\{U_{\pi}, \alpha_{[1],\pi},\ldots,\alpha_{[m],\pi}\}$. If
\begin{equation}\label{hyp1}
\alpha_{[i],P} \geq \alpha_{[i],\pi}, \quad \forall  i=1,\ldots,m,
\end{equation}
then,
\begin{equation}\label{law6}
U_P' \geq U_{\pi}',  \alpha_{[i],P}' \geq \alpha_{[i],\pi}', \quad \forall i=1,\ldots,m.
\end{equation}
\end{lemma}

\begin{proof}
Since the packet with generation time $\alpha_{[l],P}$ is delivered under policy $P$, the packet with generation time $\alpha_{[l],\pi}$ is delivered under policy $\pi$, and $\alpha_{[l],P}\geq \alpha_{[l],\pi}$, we get
\begin{equation}\label{lem3_1}
\begin{split}
 U_P'= \alpha_{[l],P}\geq \alpha_{[l],\pi}=U_{\pi}'.
 \end{split}
\end{equation}
This, together with Lemma \ref{lem3'}, implies
\begin{equation}\label{lem3_6}
\begin{split}
&\alpha_{[i],P}'\geq \alpha_{[i],\pi}', \quad  i=1,\ldots,m.
 \end{split}
\end{equation}
Hence, (\ref{law6}) holds for any queue size $B\geq 0$, which completes the proof.
\end{proof}

\begin{lemma}\label{lem4}
Suppose that under policy $P$, $\{U_P', \alpha_{[1],P}',\ldots,\alpha_{[m],P}'\}$ is obtained by adding a packet to the system whose state is $\{U_P, \alpha_{[1],P},\ldots,\alpha_{[m],P}\}$. Further, suppose that under policy $\pi$, $\{U_{\pi}', \alpha_{[1],\pi}',\ldots,\alpha_{[m],\pi}'\}$ is obtained by adding a packet to the system whose state is $\{U_{\pi}, \alpha_{[1],\pi},\ldots,\alpha_{[m],\pi}\}$. If
\begin{equation}\label{hyp2}
U_P\geq U_{\pi},
\end{equation}
then
\begin{equation}\label{law2}
U_P' \geq U_{\pi}', \alpha_{[i],P}' \geq \alpha_{[i],\pi}', \quad \forall i=1,\ldots,m.
\end{equation}
\end{lemma}

\begin{proof}
Since there is no packet delivery, we have
\begin{equation}\label{law4}
\begin{split}
 U_P'=U_P\geq U_{\pi}=U_{\pi}'.
 \end{split}
\end{equation}
This, together with Lemma \ref{lem3'}, implies
\begin{equation}\label{lem3_6}
\begin{split}
&\alpha_{[i],P}'\geq \alpha_{[i],\pi}', \quad  i=1,\ldots,m.
 \end{split}
\end{equation} 
Hence, (\ref{law2}) holds for any queue size $B\geq 0$, which completes the proof.
\end{proof}

\begin{proof}[ Proof of Lemma \ref{lem2}]
For any sample path, we have that $U_P (0^-) = U_{\pi}(0^-)$ and $\alpha_{[i],P}(0^-) = \alpha_{[i],\pi}(0^-)$ for $i=1,\ldots,m$. According to the coupling between the system state processes $\{\mathbf{V}_P(t),  t\in[0,\infty)\}$ and $\{\mathbf{V}_\pi(t),  t\in[0,\infty)\}$, as well as Lemma \ref{lem3} and \ref{lem4}, we get
\begin{equation}
\begin{split}
[U_P(t)\vert\mathcal{I}] \geq [U_{\pi}(t)\vert\mathcal{I}], [\alpha_{[i],P}(t)\vert\mathcal{I}] \geq [\alpha_{[i],\pi}(t)\vert\mathcal{I}],\nonumber
\end{split}
\end{equation}
holds for all $t\in[0,\infty)$ and $i=1,\ldots,m$. Hence, \eqref{main_eq} follows which implies \eqref{law9} by Proposition \ref{Theorem_6.B.30}.
This completes the proof.
\end{proof}

\begin{proof}[Proof of Theorem \ref{thm1}]
As a result of Lemma \ref{lem2}, we have
\begin{equation*}
\begin{split}
[\{U_{P}(t),  t\in[0,\infty)\}\vert\mathcal{I}]\geq_{\text{st}} [\{U_{\pi}(t), t\in[0,\infty)\}\vert\mathcal{I}],
 \end{split}
\end{equation*}
holds for all work-conserving policies $\pi$, which implies
\begin{equation}
\begin{split}
[\{\Delta_{P}(t), t\in[0,\infty)\}\vert\mathcal{I}]\!\!\leq_{\text{st}} \!\![\{\Delta_{\pi}(t), t\in[0,\infty)\}\vert\mathcal{I}],
 \end{split}
\end{equation}
holds for all work-conserving policies $\pi$.

For non-work-conserving policies, since the packet service times are \emph{i.i.d.} exponentially distributed, service idling only increases the waiting time of the packet in the system. Therefore, the age under non-work-conserving policies will be greater. As a result, we have 
\begin{equation}
\begin{split}
[\{\Delta_{P}(t),  t\in\!\![0,\infty)\}\vert\mathcal{I}]\leq_{\text{st}} [\{\Delta_{\pi}(t), t\in\!\![0,\infty)\}\vert\mathcal{I}], ~\forall \pi\in\Pi_{r}.\nonumber
 \end{split}
\end{equation}

Finally, \eqref{thm1eq2} follows directly from \eqref{thm1eq1} using the properties of stochastic ordering \cite{shaked2007stochastic}. This completes the proof.
\end{proof}

\section{Proof of Lemma \ref{lowerbound_lemma}}\label{Appendix_B}
The proof of Lemma \ref{lowerbound_lemma} is motivated by the proof idea of \cite[Lemma 1]{sun2016delay}. For notation simplicity, let policy $P$ represent the non-prmp-LGFS-R policy.  We need to define the following parameters:  

Define $\Gamma_{i}$ and $D_{i}$ as
\begin{equation}\label{Def_na1_1}
\begin{split}
\Gamma_{i}=\min\{v_{j}: s_j\geq s_i\},
\end{split}
\end{equation}
\begin{equation}\label{Def_na1_2}
D_{i}=\min\{c_{j}: s_j\geq s_i\}.
\end{equation}
where $\Gamma_{i}$ and $D_{i}$ are the smallest assignment time and completion time, respectively, of all packets that have generation times greater than that of packet $i$.
\begin{figure}
\centering \includegraphics[scale=0.35]{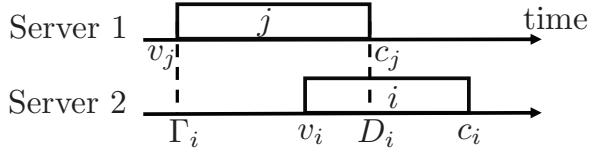}
\centering
\captionsetup{justification=justified, font={small,onehalfspacing}}
\caption{ An illustration of $v_i$, $c_i$, $\Gamma_i$, and $D_i$. There are 2 servers, and $s_j>s_i$. There is no packet with generation time greater than $s_i$ that is assigned to any of the servers before time $v_j$. Packet $j$ is assigned to Server 1 at time $v_j$ and delivered to the destination at time $c_j$; while packet $i$ is assigned to Server 2 at time $v_i$ and delivered to the destination at time $c_i$. The service starting time and completion time of packet $j$ are earlier than those of packet $i$. Thus, we have $\Gamma_i=v_j$ and $D_i=c_j$.}\label{Fig:parameter}
\vspace{-0.5cm}
\end{figure}
An illustration of these parameters is provided in Fig. \ref{Fig:parameter}. Suppose that there are $n$ update packets, where $n$ is an  arbitrary  positive  integer,  no  matter  finite  or  infinite. Define the vectors $\mathbf{\Gamma}=(\Gamma_{1}, \ldots, \Gamma_{n})$, and $\mathbf{D}=(D_{1}, \ldots, D_{n})$. All these quantities are functions of the scheduling policy $\pi$. 

Notice that we can deduce from \eqref{age} that the age process $\{\Delta_{\pi}(t), t\in [0,\infty)\}$ under any policy $\pi$ is an increasing function of $\mathbf{D}(\pi)$. Moreover, we can deduce from \eqref{lowerbound} that the process $\{\Delta_{P}^{\text{LB}}(t), t\in [0,\infty)\}$ is an increasing function of $\mathbf{\Gamma}(P)$. According to Proposition \ref{Theorem_6.B.16.(a)}, if we can show
\begin{equation}\label{eq1pflema1thm3}
[\mathbf{\Gamma}(P)\vert\mathcal{I}]\leq_{\text{st}} [\mathbf{D}(\pi)\vert\mathcal{I}],
\end{equation}
holds for all $\pi\in\Pi_{m}$, then \eqref{hypo_na2} is proven. Hence, \eqref{eq1pflema1thm3} is what we need to show. We pick an arbitrary policy $\pi\in\Pi_{m}$ and prove \eqref{eq1pflema1thm3} using Proposition \ref{Theorem_6.B.3} into two steps. 

\emph{Step 1}: Consider packet 1. Define $i^*=\argmin_i a_i$, where $s_{i^*}\geq s_1$. Since all servers are idle by time $a_{i^*}$ and policy $P$ is work-conserving policy, packet $i^*	$ will be assigned to a server under policy $P$ once it arrives. Thus, from \eqref{Def_na1_1}, we obtain
\begin{equation}\label{eq1nbu}
[\Gamma_1(P)\vert\mathcal{I}]=[v_{i^*}(P)\vert\mathcal{I}]=a_{i^*}.
\end{equation}
Under policy $\pi$, the completion times of all packets must be no smaller than $a_{i^*}$. Hence, we have
\begin{equation}
[c_{i}(\pi)\vert\mathcal{I}]\geq a_{i^*},~\forall i\geq 1.
\end{equation}
This with \eqref{Def_na1_2} imply
\begin{equation}\label{eq2nbu}
[D_{1}(\pi)\vert\mathcal{I}]\geq a_{i^*}.
\end{equation}
Combining \eqref{eq1nbu} and \eqref{eq2nbu}, we get 
\begin{equation}\label{NBU_proven1}
\begin{split}
[\Gamma_1(P)\vert\mathcal{I}]\leq [D_{1}(\pi)\vert\mathcal{I}].
\end{split}
\end{equation}

\emph{Step 2}: Consider a packet $j$, where $2\leq j \leq n$. We suppose that there is no packet with generation time greater than $s_j$ that has been delivered before packet $j$ under policy $\pi$.
We need to prove that 
\begin{equation}\label{NBU_to_proof}
\begin{split}
[\Gamma_{j}(P)\vert\mathcal{I}, \Gamma_{1}(P)=\gamma_1, \ldots, \Gamma_{j-1}(P)=\gamma_{j-1}]\\ 
\leq_{\text{st}}[D_{j}(\pi)\vert\mathcal{I}, D_{1}(\pi)=d_1, \ldots, D_{j-1}(\pi)=d_{j-1}]\\
\text{whenever}\quad \gamma_i\leq d_i, i=1, 2, \ldots, j-1. 
\end{split}
\end{equation}
For notational simplicity, define $\Gamma^{j-1}\triangleq \{\Gamma_{1}(P)=\gamma_1, \ldots, \Gamma_{j-1}(P)=\gamma_{j-1}\}$ and $D^{j-1}\triangleq \{D_{1}(\pi)=d_1, \ldots, D_{j-1}(\pi)=d_{j-1}\}$. We will show that there is at least one server under policy $P$ that can serve a new packet at a time that is stochastically smaller than the completion time of packet $j$ under policy $\pi$. At this time, there are two possible cases under policy $P$. One of them is that the idle server processes a packet with generation time greater than $s_j$. The other one is that the idle server processes a packet with generation time less than $s_j$ or there is no packet to be processed. We will show that \eqref{NBU_to_proof} holds in either case.
\begin{figure}
\includegraphics[scale=0.27]{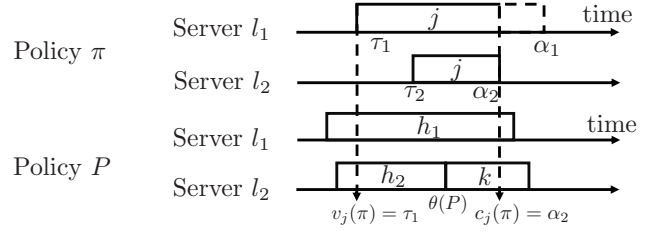}
\centering
\captionsetup{justification=justified, font={small,onehalfspacing}}
\caption{Illustration of packet assignments under policy $\pi$ and policy $P$. In policy $\pi$, two copies of packet $j$ are replicated on the server $l_1$ and server $l_2$ at time $\tau_1$ and $\tau_2$, where $v_j(\pi)=\min\{\tau_1,\tau_2\}=\tau_1$. Server $l_2$ completes one copy of packet $j$ at time $c_j(\pi)=\alpha_2$, server $l_1$ cancels its redundant copy of packet $j$ at time $c_j(\pi)$. Hence, the service duration of packet $j$ is $[v_j(\pi), c_j(\pi)]$ in policy $\pi$. In policy $P$, at least one of the servers $l_1$ and $l_2$ becomes idle before time $c_j(\pi)$. In this example, server $l_2$ becomes idle at time $\theta(P)<c_j(\pi)$ and a fresh packet $k$ with $s_k\geq s_j$ starts its service on server $l_2$ at time $\theta(P)$.}\label{Fig:Illustration1}
\vspace{-0.3cm}
\end{figure}

As illustrated in Fig. \ref{Fig:Illustration1}, suppose that $u$ copies of packet $j$ are replicated on the servers $l_1,l_2,\ldots,l_u$ at the time epochs $\tau_1, \tau_2,\ldots,\tau_u$ in policy $\pi$, where $v_{j}(\pi)=\min_{w=1,\ldots,u}\tau_w$.\footnote{If $u= 1$, then either there is no replication or policy $\pi$ decides not to replicate packet $j$.} In addition, suppose that server $l_w$ will complete serving its copy of packet $j$ at time $\alpha_w$ if there is no cancellation. Then, one of these $u$ servers will complete one copy of packet $j$ at time $c_j(\pi)=\min_{w=1,\ldots,u}\alpha_w$, which is the earliest among these $u$ servers. Hence, packet $j$ starts service at time $v_{j}(\pi)$ and completes service at time $c_j(\pi)$ in policy $\pi$. In policy $P$, let $h_w$ represent the index of the last packet that has been assigned to server $l_w$ before time $\tau_w$. Suppose that under policy $P$, server $l_w$ has spent $\chi_{l_w}$ $(\chi_{l_w}\geq 0)$ seconds on serving packet $h_w$ before time $\tau_w$. Let $R_{l_w}$ denote the remaining service time of server $l_w$ for serving packet $h_w$ after time $\tau_w$ in policy $P$. Let $X^{\pi}_{l_w}=\alpha_w-\tau_w$  denote the service time of one copy of packet $j$ in server $l_w$ under policy $\pi$ and $X^{P}_{l_w}=\chi_{l_w}+R_{l_w}$ denote the service time of packet $h_w$ in server $l_w$ under policy $P$. 
The CCDF of $R_{l_w}$ is given by
 \begin{equation}\label{remaining_service}
\begin{split}
\mathbb{P}[R_{l_w}> s]=\mathbb{P}[X^{P}_{l_w}-\chi_{l_w}>s\vert X^{P}_{l_w}>\chi_{l_w}].
\end{split}
\end{equation}
Because the packet service times are NBU, we can obtain that for all $s,\chi_{l_w}\geq 0$
 \begin{equation}\label{NBUproperty1}
\begin{split}
&\mathbb{P}[X^{P}_{l_w}-\chi_{l_w}>s\vert X^{P}_{l_w}>\chi_{l_w}]=\\&\mathbb{P}[X^\pi_{l_w}-\chi_{l_w}>s\vert X^\pi_{l_w}>\chi_{l_w}]\leq \mathbb{P}[X^\pi_{l_w}>s].
\end{split}
\end{equation}
By combining \eqref{remaining_service} and \eqref{NBUproperty1}, we obtain
\begin{equation}\label{R-X}
\begin{split}
R_{l_w}\leq_{\text{st}}X^\pi_{l_w}.
\end{split}
\end{equation}
Because the packet service times are independent across the servers, by Lemma 13 of \cite{sun2016delay}, $R_{l_1}, \ldots, R_{l_u}$ are mutually independent. By Proposition \ref{Theorem_6.B.16.(b)} and \eqref{R-X}, we can obtain
\begin{equation}\label{each_server_completion}
\min_{w=1,\ldots,u}\tau_w+R_{l_w}\leq_{\text{st}}\min_{w=1,\ldots,u}\tau_w+X^\pi_{l_w}=\min_{w=1,\ldots,u}\alpha_w.
\end{equation}
From \eqref{each_server_completion} we can deduce that at least one of the servers $l_1,\ldots,l_u$, say server $l_z$, becomes available to serve a new packet under policy $P$ at a time that is stochastically smaller than the time $c_j(\pi)=\min_{w=1,\ldots,u}\alpha_w$. Let $\theta(P)$ denote the time that server $l_z$ becomes available to serve a new packet in policy $P$. According to \eqref{each_server_completion}, we have 
 \begin{equation}
\begin{split}
[\theta(P)\vert \mathcal{I}, \Gamma^{j-1}]\leq_{\text{st}}[c_{j}(\pi)\vert \mathcal{I}, D^{j-1}]\\ \text{whenever}\quad \gamma_i\leq d_i, i=1, 2, \ldots, j-1.
\end{split}
\end{equation}
 At time $\theta(P)$, we have two possible cases under policy $P$:
 \begin{figure*}[!tbp]
 \centering
 \subfigure[Case 1: Packet $k$ is assigned to server $l_z$ after the completion of packet $h_z$.]{
  \includegraphics[scale=0.28]{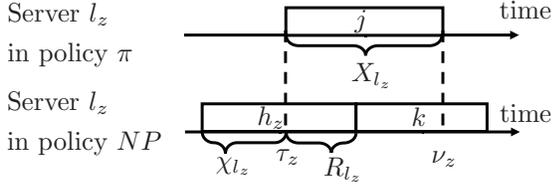}
   \label{a}
   }
 \subfigure[Case 2: Packet $k$ is assigned to server $l'$ before the completion of packet $h_z$.]{
  \includegraphics[scale=0.28]{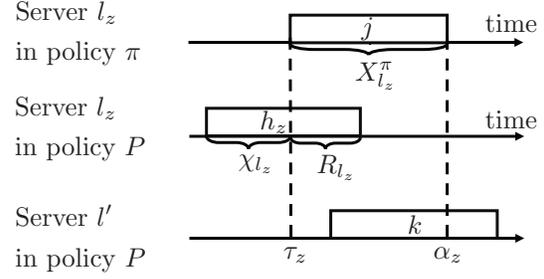}
   \label{b}
   }
   \caption{The possible cases to occur after the completion of packet $h_z$.}
   \label{fig:samp_path_NBU}
\end{figure*}

Case 1: A fresh packet $k$ is assigned at time $\theta(P)$ to server $l_z$ under policy $P$ such that $s_k\geq s_j$, as shown in Fig. \ref{a}. Hence, we obtain
\begin{equation}\label{NBUeq2}
\begin{split}
[v_{k}(P)\vert \mathcal{I}, \Gamma^{j-1}]=[\theta(P)\vert \mathcal{I}, \Gamma^{j-1}]\leq_{\text{st}}[c_{j}(\pi)\vert  \mathcal{I},D^{j-1}]\\ \text{whenever}\quad \gamma_i\leq d_i, i=1, 2, \ldots, j-1.
\end{split}
\end{equation}
Since $s_k\geq s_j$, \eqref{Def_na1_1} implies
\begin{equation}\label{NBUeq2'}
[\Gamma_{j}(P)\vert\mathcal{I}, \Gamma^{j-1}]\leq [v_{k}(P)\vert \mathcal{I}, \Gamma^{j-1}]
\end{equation}
Since there is no packet with generation time greater than $s_j$ that has been delivered before packet $j$ under policy $\pi$, \eqref{Def_na1_2} implies
\begin{equation}\label{NBUeq2''}
[D_{j}(\pi)\vert\mathcal{I}, D^{j-1}]= [c_{j}(\pi)\vert \mathcal{I}, D^{j-1}]
\end{equation}
 By combining \eqref{NBUeq2}, \eqref{NBUeq2'}, and \eqref{NBUeq2''}, \eqref{NBU_to_proof} follows.

Case 2: A packet with generation time smaller than $s_j$ is assigned to server $l_z$ or there is no packet assignment to server $l_z$ at time $\theta(P)$ under policy $P$. Since policy $P$ is a work-conserving policy, policy $P$ serves the packet with the largest generation time first, and the packet generation times $(s_1, s_2, \ldots)$ and arrival times $(a_1,a_2, \ldots)$ are invariant of the scheduling policy, a packet $k$ with $s_k\geq s_j$ must have been assigned to another server, call it server $l'$, before time $\theta(P)$, as shown in Fig. \ref{b}. Hence, we obtain
\begin{equation}\label{NBUeq3}
\begin{split}
[v_{k}(P)\vert \mathcal{I}, \Gamma^{j-1}]\leq [\theta(P)\vert  \mathcal{I},\Gamma^{j-1}]\leq_{\text{st}}[c_{j}(\pi)\vert \mathcal{I}, D^{j-1}]\\ \text{whenever}\quad \gamma_i\leq d_i, i=1, 2, \ldots, j-1.
\end{split}
\end{equation}
Similar to Case 1, we can use \eqref{Def_na1_1}, \eqref{Def_na1_2}, and \eqref{NBUeq3} to show that \eqref{NBU_to_proof} follows in this case.


It is important to note that if there is a packet $y$ with $s_y>s_j$ and $c_y(\pi)<c_j(\pi)$ (this may occur if packet $y$ preempts the service of packet $j$ under policy $\pi$ or packet $y$ arrives to the system before packet $j$), then we replace packet $j$ by packet $y$ in the arguments and equations from \eqref{NBU_to_proof} to \eqref{NBUeq3} to obtain
 \begin{equation}\label{packet_p1}
\begin{split}
[\Gamma_y(P)\vert \mathcal{I}, \Gamma^{j-1}]\leq_{\text{st}}[D_{y}(\pi)\vert \mathcal{I}, D^{j-1}]\\ \text{whenever}\quad \gamma_i\leq d_i, i=1, 2, \ldots, j-1.
\end{split}
\end{equation}
Observing that $s_y>s_j$, \eqref{Def_na1_1} implies
 \begin{equation}\label{packet_p3}
\begin{split}
[\Gamma_j(P)\vert \mathcal{I}, \Gamma^{j-1}]\leq[\Gamma_y(P)\vert \mathcal{I}, \Gamma^{j-1}].
\end{split}
\end{equation}
Since $c_y(\pi)<c_j(\pi)$ and $s_y>s_j$, \eqref{Def_na1_2} implies
\begin{equation}\label{packet_p2}
\begin{split}
[D_{j}(\pi)\vert \mathcal{I}, D^{j-1}]=[D_{y}(\pi)\vert \mathcal{I}, D^{j-1}].
\end{split}
\end{equation}
By combining \eqref{packet_p1}, \eqref{packet_p3}, and \eqref{packet_p2}, we can prove  \eqref{NBU_to_proof} in this case too. Now, substituting \eqref{NBU_proven1} and \eqref{NBU_to_proof} into Proposition \ref{Theorem_6.B.3}, \eqref{eq1pflema1thm3} is proven. This completes the proof.

\section{Proof of Theorem \ref{thmNBU1}}\label{Appendix_D}
For notation simplicity, let policy $P$ represent the non-prmp-LGFS-R policy. 
\begin{proof}[Proof of Theorem \ref{thmNBU1}.(a)]
We prove Theorem \ref{thmNBU1}.(a) into two steps:

\emph{Step 1}: We will show that the average gap between $\Delta_{P}^{\text{LB}}$ and $\Delta_{P}$ is upper bounded by $\mathbb{E}[X]$. Recall the definitions of $\Gamma_i$ and $D_i$ from \eqref{Def_na1_1} and \eqref{Def_na1_2}, respectively. Define $d_i(P)=D_i(P)-\Gamma_i(P)$. We know that there is a packet $k$ with $s_k\geq s_i$ that starts service at time $\Gamma_i(P)$ under policy $P$. Without loss of generality, suppose that a server $l$ is processing a copy of packet $k$. Because of replications, packet $k$ completes service under policy $P$ as soon as one of its replica completes service. Hence, packet $k$ is delivered at time $c_k(P)$ that is no later than $\Gamma_i(P)+X_l$ under policy $P$. This implies that $c_k(P)-\Gamma_i(P)\leq X_l$. From \eqref{Def_na1_2}, we can deduce that $D_i(P)-\Gamma_i(P)\leq c_k(P)-\Gamma_i(P)\leq X_l$. From this, we can obtain
\begin{equation}\label{dl_eq}
\mathbb{E}[d_i]\leq \mathbb{E}[X],~\forall i.
\end{equation}
We now proceed to characterize the gap between $\Delta_{P}^{\text{LB}}$ and $\Delta_{P}$. We use $\{G(t), t\in [0,\infty)\}$ to denote the gap process between $\Delta_{P}^{\text{LB}}$ and $\Delta_{P}$. The average gap is given by
\begin{equation}\label{gapeq1'}
[\bar{G}\vert\mathcal{I}]=\limsup_{T\rightarrow\infty}\frac{\mathbb{E}[\int_0^{T} G(t)dt]}{T}.
\end{equation}
Let $\tau_{i}$ denote the inter-generation time between packet $i$ and packet $i-1$ (i.e., $\tau_{i}=s_{i}-s_{i-1}$), where $\mathbf{\tau}=\{\tau_i, i\geq 1\}$. Note that, since the packet service times are independent of the packet generation process, we have $d_i$'s are independent of $\mathbf{\tau}$. Define $N(T)=\max\{i : s_i\leq T\}$ as the number of generated packets by time $T$. Note that $[0,s_{N(T)}]\subseteq [0,T]$, where the length of the interval $[0,s_{N(T)}]$ is $\sum_{i=1}^{N(T)}\tau_i$. Thus, we have
\begin{equation}\label{upper1}
\sum_{i=1}^{N(T)}\tau_i\leq T.
\end{equation}
The area defined by the integral in \eqref{gapeq1'} can be decomposed into a sum of disjoint geometric parts. Observing Fig.  \ref{Fig:Gap1}, the area can be approximated to the concatenation of the parallelograms $G_1, G_2, \ldots$ ($G_i$'s are highlighted in Fig. \ref{Fig:Gap1}). Note that the parallelogram $G_i$ results after the generation of packet $i$ (i.e., the gap that is corresponding to the packet $i$ occurs after its generation). Since the observing time $T$ is chosen arbitrary, when $T\geq s_{i}$, the total area of the parallelogram $G_i$ is accounted in the summation $\sum_{i=1}^{N(T)}G_i$, while it may not be accounted in the integral $\int_0^{T} G(t)dt$. This implies that 
\begin{equation}\label{upper2}
\sum_{i=1}^{N(T)}G_i\geq\int_0^{T} G(t)dt.
\end{equation}
Combining \eqref{upper1} and \eqref{upper2}, we get
\begin{equation}\label{gapeq1}
\frac{\int_0^{T} G(t)dt}{T}\leq\frac{\sum_{i=1}^{N(T)}G_i}{\sum_{i=1}^{N(T)}\tau_i}.
\end{equation} 
\begin{figure}
\includegraphics[scale=0.35]{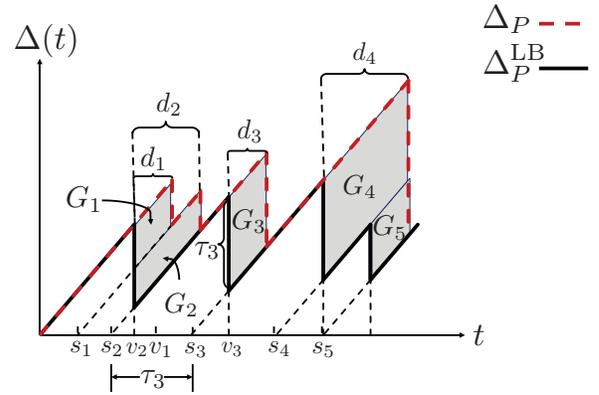}
\centering
\caption{The evolution of $\Delta_{P}^{\text{LB}}$ and $\Delta_{P}$ in a queue with 4 servers and $r=2$.}\label{Fig:Gap1}
\vspace{-0.3cm}
\end{figure}
Then, take conditional expectation given $\mathbf{\tau}$ and $N(T)$ on both sides of \eqref{gapeq1}, we obtain
\begin{equation}\label{gapeq3}
\begin{split}
&\frac{\mathbb{E}[\int_0^{T} G(t)dt\vert\mathbf{\tau}, N(T)]}{T}\leq \\&\frac{\mathbb{E}[\sum_{i=1}^{N(T)}G_i\vert\mathbf{\tau}, N(T)]}{\sum_{i=1}^{N(T)}\tau_i}=\frac{\sum_{i=1}^{N(T)}\mathbb{E}[G_i\vert\mathbf{\tau}, N(T)]}{\sum_{i=1}^{N(T)}\tau_i},
\end{split}
\end{equation}
where the second equality follows from the linearity of the expectation. From Fig. \ref{Fig:Gap1}, $G_i$ can be calculated as
\begin{equation}\label{gapeq2'}
G_i=\tau_{i}d_i.
\end{equation}
Substituting by \eqref{gapeq2'} into \eqref{gapeq3}, yields
\begin{equation}\label{gapeq4}
\begin{split}
&\frac{\mathbb{E}[\int_0^{T} G(t)dt\vert\mathbf{\tau}, N(T)]}{T}\leq\\&\frac{\sum_{i=1}^{N(T)}\mathbb{E}[\tau_{i}d_i\vert\mathbf{\tau}, N(T)]}{\sum_{i=1}^{N(T)}\tau_i}=\frac{\sum_{i=1}^{N(T)}\tau_{i}\mathbb{E}[d_i\vert\mathbf{\tau}, N(T)]}{\sum_{i=1}^{N(T)}\tau_i}.
\end{split}
\end{equation}
Note that $d_i$'s are independent of $\mathbf{\tau}$. Thus, we have $\mathbb{E}[d_i\vert\mathbf{\tau}, N(T)]=\mathbb{E}[d_i]\leq\mathbb{E}[X]$ for all $i$. Substituting this into \eqref{gapeq4}, yields
\begin{equation}\label{gapeq5}
\frac{\mathbb{E}[\int_0^{T} G(t)dt\vert\mathbf{\tau}, N(T)]}{T}\leq\frac{\sum_{i=1}^{N(T)}\tau_{i}\mathbb{E}[X]}{\sum_{i=1}^{N(T)}\tau_i}=\mathbb{E}[X],
\end{equation}
by the law of iterated expectations, we have
\begin{equation}\label{gap-char'}
\frac{\mathbb{E}[\int_0^{T} G(t)dt]}{T}\leq\mathbb{E}[X].
\end{equation}
Taking the $\limsup$ of both side of \eqref{gap-char'} when $T\rightarrow\infty$, yields
\begin{equation}\label{gap-char}
\limsup_{T\rightarrow\infty}\frac{\mathbb{E}[\int_0^{T} G(t)dt]}{T}\leq\mathbb{E}[X].
\end{equation}
Equation \eqref{gap-char} tells us that the average gap between $\Delta_{P}^{\text{LB}}$ and $\Delta_{P}$ is no larger than $\mathbb{E}[X]$.

\emph{Step 2}: We prove \eqref{NBUhypoeq1}. Since $\Delta_{P}^{\text{LB}}$ is a lower bound of the age process of policy $P$ and the average gap between $\Delta_{P}^{\text{LB}}$ and $\Delta_{P}$ is no larger than $\mathbb{E}[X]$, we obtain 
\begin{align}\label{lowerbound-relation1}
[\bar{\Delta}_{P}^{\text{LB}}\vert\mathcal{I}]\leq [\bar{\Delta}_{P}\vert\mathcal{I}]\leq [\bar{\Delta}_{P}^{\text{LB}}\vert\mathcal{I}] +\mathbb{E}[X],
\end{align}
where $\bar{\Delta}_{P}^{\text{LB}}=\limsup_{T\rightarrow\infty}\frac{\mathbb{E}[\int_0^{T} \Delta_{P}^{\text{LB}}(t)dt]}{T}$. From Lemma \ref{lowerbound_lemma}, we have for all $\mathcal{I}$ satisfying $B\geq 1$, and $\pi\in\Pi_{m}$
\begin{equation}
[\{\Delta_{P}^{\text{LB}}(t), t\in[0,\infty)\}\vert\mathcal{I}]\leq_{\text{st}}[\{\Delta_\pi(t), t\in [0,\infty)\}\vert\mathcal{I}],
\end{equation}
which implies that 
\begin{equation}
[\bar{\Delta}_{P}^{\text{LB}}\vert\mathcal{I}]\leq[\bar{\Delta}_{\pi}\vert\mathcal{I}],
\end{equation}
holds for all $\pi\in\Pi_{m}$. As a result, we get
\begin{equation}\label{thmnbu_1R}
[\bar{\Delta}_{P}^{\text{LB}}\vert\mathcal{I}]\leq\min_{\pi\in\Pi_m}[\bar{\Delta}_{\pi}\vert\mathcal{I}].
\end{equation}
Since policy non-prmp-LGFS-R is a feasible policy, we get
\begin{equation}\label{thmnbu_1RR}
\min_{\pi\in\Pi_m} [\bar{\Delta}_{\pi}\vert\mathcal{I}]\leq [\bar{\Delta}_{P}\vert\mathcal{I}].
\end{equation}
Combining  \eqref{lowerbound-relation1}, \eqref{thmnbu_1R}, and \eqref{thmnbu_1RR}, we get
\begin{equation}
\min_{\pi\in\Pi_m} [\bar{\Delta}_{\pi}\vert\mathcal{I}]\leq [\bar{\Delta}_{P}\vert\mathcal{I}]\leq \min_{\pi\in\Pi_m} [\bar{\Delta}_{\pi}\vert\mathcal{I}]+\mathbb{E}[X],
\end{equation} 
which completes the proof. 
\end{proof}

\begin{proof}[Proof of Theorem \ref{thmNBU1}.(b)]
The proof of part (b) is similar to that of part (a). Define $d_i(P)=D_i(P)-\Gamma_i(P)$. We know that there is a packet $k$ with $s_k\geq s_i$ that starts service at time $\Gamma_i(P)$ under policy $P$. Since $m=ar$ for a positive integer $a$, packet $k$ is processed by $r$ servers in policy $P$. Let $\mathcal{S}_{k}\subseteq \{1,\ldots,m\}$ be the set of servers that process packet $k$ under policy $P$, which satisfies $\vert\mathcal{S}_{k}\vert=r$.  Because of replications, packet $k$ completes service under policy $P$ as soon as one of its replica is completes service. Hence, packet $k$ is delivered at time $c_k(P)=\Gamma_{i}(P) + \min_{l\in\mathcal{S}_{k}}X_{l}$ under policy $P$. This implies that $c_k(P)-\Gamma_{i}(P)= \min_{l\in\mathcal{S}_{k}}X_{l}$. From \eqref{Def_na1_2}, We can deduce that $D_i(P)-\Gamma_i(P)\leq c_k(P)-\Gamma_{i}(P)= \min_{l\in\mathcal{S}_{k}}X_{l}$. From this, we can obtain
\begin{equation}
\mathbb{E}[d_i]\leq\mathbb{E}[\min_{l=1,\ldots,r}X_{l}],\forall i.
\end{equation}
Similar to part a, we use $\{G(t), t\in [0,\infty)\}$ to denote the gap process between $\Delta_{P}^{\text{LB}}$ and $\Delta_{P}$. The average gap is given by
\begin{equation}
[\bar{G}\vert\mathcal{I}]=\limsup_{T\rightarrow\infty}\frac{\mathbb{E}[\int_0^{T} G(t)dt]}{T}.
\end{equation}
Following the same steps as in the proof of part (a), we can show that
\begin{equation}\label{gap-charR}
\limsup_{T\rightarrow\infty}\frac{\mathbb{E}[\int_0^{T} G(t)dt]}{T}\leq\mathbb{E}[\min_{l=1,\ldots,r}X_{l}].
\end{equation}
Equation \eqref{gap-charR} tells us that the average gap between $\Delta_{P}^{\text{LB}}$ and $\Delta_{P}$ is no larger than $\mathbb{E}[\min_{l=1,\ldots,r}X_{l}]$. This and the fact that  $\Delta_{P}^{\text{LB}}$ is a lower bound of the age process of policy $P$, imply
\begin{align}\label{lowerbound-relation1b}
[\bar{\Delta}_{P}^{\text{LB}}\vert\mathcal{I}]\leq [\bar{\Delta}_{P}\vert\mathcal{I}]\leq [\bar{\Delta}_{P}^{\text{LB}}\vert\mathcal{I}] +\mathbb{E}[\min_{l=1,\ldots,r}X_{l}].
\end{align} 
Similar to part (a), we can use \eqref{lowerbound-relation1b} with Lemma \ref{lowerbound_lemma} to show that
\begin{equation}
\min_{\pi\in\Pi_m} [\bar{\Delta}_{\pi}\vert\mathcal{I}]\leq [\bar{\Delta}_{P}\vert\mathcal{I}]\leq \min_{\pi\in\Pi_m} [\bar{\Delta}_{\pi}\vert\mathcal{I}]+\mathbb{E}[\min_{l=1,\ldots,r}X_{l}],
\end{equation} 
which completes the proof. 
\end{proof}

\section{Proof of Theorem \ref{thm2th}}\label{Appendix_C}
We follow the same proof technique of Theorem \ref{thm1}. We start by comparing policy $P$ (prmp-LGFS-R policy) with an arbitrary work-conserving policy $\pi$. For this, we need to define the system state of any policy $\pi$:

\begin{definition}  At any time $t$, the system state of policy $\pi$ is specified by  $H_{\pi}(t)=( N_\pi(t), \gamma_\pi(t))$, where $N_\pi(t)$ is the total number of distinct packets in the system at time $t$ (excluding packet replicas). Define $\gamma_\pi(t)$ as the total number of distinct packets that are delivered to the destination at time $t$. Let $\{H_{\pi}(t), t\in[0,\infty)\}$ be the state process of policy $\pi$, which is assumed to be right-continuous. 
\end{definition}

To prove Theorem \ref{thm2th}, we will need the following lemma.
\begin{lemma}\label{lem5}
For any work-conserving policy $\pi$, if $H_{P}(0^-)=H_{\pi}(0^-)$ and $B=\infty$, then $[\{{H}_{P}(t),  t\in[0,\infty)\}\vert\mathcal{I}]$ and $[\{H_{\pi}(t), t\in[0,\infty)\}\vert\mathcal{I}]$ are of the same distribution.
\end{lemma}

 Suppose that $\{\widetilde{H}_P(t), t\in[0,\infty)\}$ and $\{\widetilde{H}_{\pi}(t), t\in[0,\infty)\}$ are stochastic processes having the same stochastic laws as $\{H_P(t), t\in[0,\infty)\}$ and $\{H_{\pi}(t), t\in[0,\infty)\}$. Now, we couple  the packet delivery times during the evolution of $\widetilde{H}_P(t)$ to be identical with the packet delivery times during the evolution of $\widetilde{H}_{\pi}(t)$. Such a coupling is valid because the service times are exponentially distributed, and hence, memoryless. 
 
To ease the notational burden, we will omit the tildes henceforth on the coupled versions and just use $\{H_P(t)\}$ and $\{H_{\pi}(t)\}$. The following two lemmas are needed to prove Lemma \ref{lem5}:

\begin{lemma}\label{lem6}
 Suppose that under policy $P$, $\{N_P', \gamma_P'\}$ is obtained by delivering a packet to the destination in the system whose state is $\{N_P, \gamma_P\}$. Further, suppose that under policy $\pi$, $\{N_\pi', \gamma_\pi'\}$ is obtained  by delivering a packet to the destination in the system whose state is $\{N_\pi, \gamma_\pi\}$. If
\begin{equation*}
N_P = N_{\pi},  \gamma_{P} = \gamma_{\pi},
\end{equation*}
then
\begin{equation}\label{law6_2}
N_P' = N_{\pi}',  \gamma_{P}' = \gamma_{\pi}'.
\end{equation}
\end{lemma}

\begin{proof}
Because the packet service times are \emph{i.i.d.} and the CCDF $\bar{F}$ is continuous, the probability for any two servers to complete their packets at the same time is zero. Therefore, in policy $P$, if one copy of a replicated packet is completed on a server, the remaining replicated copies of this packet are still being processed on the other servers; these replicated packet copies are cancelled immediately and a new packet is replicated on these servers. Since there is a packet delivery, we have
\begin{equation*}
\begin{split}
&N_P'=N_P-1=N_\pi-1=N_\pi',\\
&\gamma_P'=\gamma_P+1=\gamma_\pi+1=\gamma_\pi'.
\end{split}
\end{equation*}
Hence, (\ref{law6_2}) holds, which complete the proof.
\end{proof}

\begin{lemma}\label{lem7}
Suppose that under policy $P$, $\{N_P', \gamma_P'\}$ is obtained by adding a new packet to the system whose state is $\{N_P, \gamma_P\}$. Further, suppose that under policy $\pi$, $\{N_\pi', \gamma_\pi'\}$ is obtained by adding a new packet to the system whose state is $\{N_\pi, \gamma_\pi\}$. If
\begin{equation*}
N_P = N_{\pi},  \gamma_{P} = \gamma_{\pi},
\end{equation*}
then
\begin{equation}\label{law7_2}
N_P' = N_{\pi}',  \gamma_{P}' = \gamma_{\pi}'.
\end{equation}
\end{lemma}

\begin{proof}
Because $B=\infty$, no packet is dropped in policy $P$ and policy $\pi$.
Since there is a new added packet to the system, we have
\begin{equation*}
\begin{split}
N_P'=N_P+1=N_\pi+1=N_\pi'.
\end{split}
\end{equation*}
Also, there is no packet delivery, hence 
\begin{equation*}
\begin{split}
\gamma_P'=\gamma_P=\gamma_\pi=\gamma_\pi'.
\end{split}
\end{equation*}
Thus, (\ref{law7_2}) holds, which complete the proof.
\end{proof}

\begin{proof}[ Proof of Lemma \ref{lem5}] 
For any sample path, we have that $N_P (0^-) = N_{\pi}(0^-)$ and $\gamma_{P}(0^-) = \gamma_{\pi}(0^-)$. According to the coupling between the system state processes $\{H_P(t),  t\in[0,\infty)\}$ and $\{H_\pi(t),  t\in[0,\infty)\}$, as well as Lemma \ref{lem6} and \ref{lem7}, we get
\begin{equation}
\begin{split}
[N_P(t)\vert\mathcal{I}] = [N_{\pi}(t)\vert\mathcal{I}], [\gamma_{P}(t)\vert\mathcal{I}] = [\gamma_{\pi}(t)\vert\mathcal{I}], \nonumber
\end{split}
\end{equation}
holds for all $t\in[0,\infty)$. This implies that $[\{{H}_P(t),  t\in[0,\infty)\}\vert\mathcal{I}]$ and $[\{H_{\pi}(t), t\in[0,\infty)\}\vert\mathcal{I}]$ are of the same distribution, which completes the proof.
%
%
\end{proof}

\begin{proof}[Proof of Theorem \ref{thm2th}]
As a result of Lemma \ref{lem5}, $[\{\gamma_{P}(t),  t\in[0,\infty)\}\vert\mathcal{I}]$ and $[\{\gamma_\pi(t), t\in[0,\infty)\}\vert\mathcal{I}]$ are of the same distribution. This implies that policy $P$ and policy $\pi$ have the same throughput performance. Also, from Lemma \ref{lem5}, we have that $[\{N_{P}(t),  t\in[0,\infty)\}\vert\mathcal{I}]$ and $[\{N_\pi(t), t\in[0,\infty)\}\vert\mathcal{I}]$ are of the same distribution. Hence, policy $P$ and policy $\pi$ have the same delay performance. These imply that policy $P$ has the same throughput and delay performance as any work-conserving policy. 

Finally, since the service times are \emph{i.i.d.}, service idling only increases the waiting time of the packet in the system. Therefore, the throughput and delay performance under non-work-conserving policies will be worse. As a result, the prmp-LGFS-R policy is throughput-optimal and delay-optimal among all policies in $\Pi_m$ (indeed, all work-conserving policies with infinite buffer size $B=\infty$ have the same throughput and delay performance, and hence, they are throughput-optimal and delay-optimal).
\end{proof}
\bibliographystyle{IEEEbib}
\bibliography{MyLib}

\end{document}